\documentclass[lettersize,journal]{IEEEtran}
\usepackage{amsmath,amsfonts}
\usepackage{array}
\usepackage{textcomp}
\usepackage{stfloats}
\usepackage{url}
\usepackage{verbatim}
\usepackage{graphicx}
\usepackage{cite}
\usepackage{soul}
\usepackage{color}
\usepackage{algorithm}
\usepackage{algorithmic}
\usepackage{caption}
\usepackage{subcaption}
\usepackage{microtype}
\usepackage{graphicx}
\usepackage{booktabs} 

\usepackage{amsmath}
\usepackage{amssymb}
\usepackage{mathtools}
\usepackage{amsthm}
\usepackage{enumitem}

\usepackage{xr-hyper}
\usepackage{hyperref}
\usepackage[capitalize,noabbrev]{cleveref}

\makeatletter
\newcommand*{\addFileDependency}[1]{
  \typeout{(#1)}
  \@addtofilelist{#1}
  \IfFileExists{#1}{}{\typeout{No file #1.}}
}
\makeatother

\newtheorem{theorem}{Theorem}

\theoremstyle{definition}
\newtheorem{definition}{Definition}
\newtheorem*{remark}{Remark}
\newtheorem{assumption}{Assumption}
\newtheorem{lemma}{Lemma}
\newtheorem{proposition}{Proposition}

\usepackage{url}
\usepackage{graphicx,wrapfig,lipsum} 
\usepackage{times}
\usepackage{soul}
\usepackage{url}
\usepackage[utf8]{inputenc}
\usepackage{caption}
\usepackage{graphicx}
\usepackage{amsthm}
\usepackage{booktabs}
\usepackage{mwe} 
\usepackage{wrapfig}

\usepackage[ruled,linesnumbered,noend,algo2e]{algorithm2e}
\usepackage{setspace}
\usepackage{amsmath,amssymb,amsfonts}

\usepackage{amsmath,amsfonts,bm}

\def\eqref#1{(\ref{#1})}

\def\1{\bm{1}}

\def\gB{{\mathcal{B}}}

\def\gF{{\mathcal{F}}}

\def\argmax{\text{arg\,max}}
\def\argmin{\text{arg\,min}}

\crefformat{equation}{(#2#1#3)}
\crefrangeformat{equation}{(#3#1#4) to~(#5#2#6)}
\crefmultiformat{equation}{(#2#1#3)}%
{ and~(#2#1#3)}{, (#2#1#3)}{ and~(#2#1#3)}
\crefname{assumption}{assumption}{assumptions}
\crefname{Assumption}{Assumption}{Assumptions}
\usepackage{bbm}
\usepackage[mathscr]{euscript}
\usepackage{graphicx}
\usepackage{textcomp}
\usepackage{xcolor}
\usepackage[detect-all]{siunitx}

\usepackage{xspace}
\usepackage{bm}
\usepackage{array}
\usepackage{multirow}

\makeatletter
\newcommand{\multilines}[1]{%
	\begin{tabularx}{\dimexpr\linewidth-\ALG@thistlm}[t]{@{}X@{}}
		#1
	\end{tabularx}
}
\makeatother

\usepackage[textsize=tiny]{todonotes}

\newcommand{\OurAlg}{\ensuremath{\small{\textsf{WaSeCom}}}\xspace}
\Crefname{figure}{Fig.}{Figs.}
\Crefname{table}{Table}{Tables}
\Crefname{section}{Sec.}{Secs.}

\DeclarePairedDelimiterX{\norm}[1]{\lVert}{\rVert}{#1}
\DeclarePairedDelimiterX{\abs}[1]{\lvert}{\rvert}{#1}

\newcommand{\set}[1]{\left\lbrace#1\right\rbrace}

\DeclareSymbolFont{extraup}{U}{zavm}{m}{n}
\DeclareMathSymbol{\varheart}{\mathalpha}{extraup}{86}

\newcommand{\defeq}{\vcentcolon=}
\newcommand{\eqdef}{=\vcentcolon}

\newcommand{\bigS}[1]{\ensuremath{\bigl[#1\bigr]}}
\newcommand{\bigC}[1]{\ensuremath{\bigl\{#1\bigr\}}}

\newcommand{\biggS}[1]{\ensuremath{\biggl[#1\biggr]}}
\newcommand{\biggC}[1]{\ensuremath{\biggl\{#1\biggr\}}}

\newcommand{\BigP}[1]{\ensuremath{\Bigl(#1\Bigr)}}
\newcommand{\BigS}[1]{\ensuremath{\Bigl[#1\Bigr]}}
\newcommand{\BigC}[1]{\ensuremath{\Bigl\{#1\Bigr\}}}

\newcommand{\BiggS}[1]{\ensuremath{\Biggl[#1\Biggr]}}

\DeclareUnicodeCharacter{FB01}{fi}

\newcommand{\blue}[1]{{\color{black}#1}\xspace}

\newcommand{\myquad}[1][1]{\hspace*{#1em}\ignorespaces}

\begin{document}

\title{Distributionally Robust Wireless Semantic Communication with Large AI Models}

\author{Long Tan Le,
        Senura Hansaja Wanasekara,
        Zerun Niu, 
        Nguyen H. Tran,~\IEEEmembership{Senior Member,~IEEE,} \\
        Phuong Vo, 
        Walid Saad,~\IEEEmembership{Fellow,~IEEE,}
        Dusit Niyato,~\IEEEmembership{Fellow,~IEEE,}
        Zhu Han,~\IEEEmembership{Fellow,~IEEE,}
        Choong~Seon~Hong,~\IEEEmembership{Fellow,~IEEE,}
        H.~Vincent Poor,~\IEEEmembership{Life Fellow,~IEEE.}
        
\thanks{L.~T.~Le, N.~H.~Tran, S.~H.~Wanasekara, and Z.~Niu are with the School of Computer Science, The University of Sydney, Darlington, NSW 2006, Australia (email: \{long.le, nguyen.tran\}@sydney.edu.au, \{zniu9834, wwan0281\}@uni.sydney.edu.au). 

P.~Vo is with the School of Computer Science and Engineering, International University-VNUHCM, Ho Chi Minh City 70000, Vietnam
(e-mail: vtlphuong@hcmiu.edu.vn). 

W.~Saad is with the Bradley Department of Electrical and
Computer Engineering, Virginia Tech, Alexandria, VA, 22305 USA (e-mail:
walids@vt.edu).

D.~Niyato is with the College of Computing and Data Science,
Nanyang Technological University, Singapore (e-mail: dniyato@ntu.edu.sg).

Z.~Han is with the Department of Electrical and Computer Engineering, University of Houston, Houston, TX, USA (email: zhan2@mail.uh.edu).

C. S. Hong and N.~H.~Tran are with the Department of Computer Science and Engineering, School of Computing, Kyung Hee University, Yongin-si 17104, Republic of Korea (email: cshong@khu.ac.kr).

H.~V.~Poor is with the School of Engineering and Applied Science, Princeton University, Princeton, NJ 08544 USA (e-mail: poor@princeton.edu). 
}
}

\markboth{Journal of \LaTeX\ Class Files,~Vol.~14, No.~8, August~2021}%
{Shell \MakeLowercase{\textit{et al.}}: A Sample Article Using IEEEtran.cls for IEEE Journals}


\maketitle

\begin{abstract}

%

\blue{Semantic communication (SemCom) has emerged as a promising paradigm for 6G wireless systems by transmitting task-relevant information rather than raw bits, yet existing approaches remain vulnerable to dual sources of uncertainty: semantic misinterpretation arising from imperfect feature extraction and transmission-level perturbations from channel noise. Current deep learning based SemCom systems typically employ domain-specific architectures that lack robustness guarantees and fail to generalize across diverse noise conditions, adversarial attacks, and out-of-distribution data.} In this paper, a novel and generalized semantic communication framework called \OurAlg is proposed to systematically address uncertainty and enhance robustness. In particular, Wasserstein distributionally robust optimization is employed to provide resilience against semantic misinterpretation and channel perturbations. A rigorous theoretical analysis is performed to establish the robust generalization guarantees of the proposed framework. Experimental results on image and text transmission demonstrate that \OurAlg achieves improved robustness under noise and adversarial perturbations. These results highlight its effectiveness in preserving semantic fidelity across varying wireless conditions.


\end{abstract}

\begin{IEEEkeywords}
Semantic Communication, Wireless Networks, Large AI Models.
\end{IEEEkeywords}
\section{Introduction}
\label{sec:intro}

The sixth generation (6G) of wireless cellular networks must be designed to handle massive data volumes, ultra-low latency, and extensive connectivity, thus addressing the increasingly sophisticated demands of emerging applications~\cite{semcom_app}. However, traditional communication paradigms, which primarily focus on the accurate transmission of raw data bits, are becoming inadequate for effectively meeting the stringent requirements of emerging data-intensive and latency-sensitive applications. To address these challenges, the concept of semantic communication (SemCom) emerged as a novel paradigm aimed at enhancing communication efficiency by transmitting task-relevant semantic information rather than raw data~\cite{semcom_survey}. By focusing on the semantic content, i.e., the meaning and relevance of information in the context of specific tasks, SemCom can significantly reduce bandwidth requirements, mitigate latency, and enhance robustness to interference and noise~\cite{semcom_principle}. 
These characteristics make SemCom particularly promising for scenarios requiring real-time decision-making and resilient communication, including autonomous driving, remote surgery, intelligent transportation systems, and time-critical industrial automation.



Building on this foundation, the integration of machine learning, particularly deep learning, has significantly advanced semantic communication by automating the extraction, representation, and interpretation of semantic content~\cite{semcom_app}. Early deep learning-based SemCom methods like those used in~\cite{deepjscc, deepjscc_text, deepsc_st} typically adopt modality-specific architectures. Despite being effective in specialized contexts, the adaptability and generalizability of the methods in \cite{deepjscc, deepjscc_text, deepsc_st} remain limited due to reliance on domain-specific knowledge and handcrafted features. More recently, the emergence of large-scale artificial intelligence (AI) models, such as transformers~\cite{transformer} and large language models (LLMs)~\cite{llm}, has transformed the landscape of SemCom. \blue{AI techniques like GPT~\cite{gpt} and causal reasoning~\cite{casual_reasoning} allow a network to leverage vast datasets and advanced training methodologies. As a result, these models can capture complex semantic relationships across diverse data modalities effectively. Consequently, such large-scale architectures can substantially enhance the encoding and decoding accuracy and adaptability of SemCom, and thus making these advanced systems particularly valuable in wireless environments.}

Despite the promising advancements in SemCom, existing systems remain inherently vulnerable to noise and uncertainty, which pose significant challenges to their reliability in wireless networks. This vulnerability largely stems from the fact that SemCom operates on high-level semantic representations, which are more sensitive to perturbations than traditional bit-level signals. These perturbations can originate from two fundamental sources: \textit{semantic-level noise}, caused by ambiguity or errors in extracting and interpreting task-relevant meaning; and \textit{transmission-level noise}, resulting from distortions during wireless propagation~\cite{semcom_principle}. While recent efforts have explored robustness techniques to mitigate these effects, most existing solutions are developed under constrained assumptions or target specific use cases~\cite{robust_sc_1, masked_vqae, robust_sc_2, robust_sc_3, robust_sc_4} and, hence, they lack the flexibility to generalize across varying tasks, modalities, and network conditions. As a result, current SemCom systems remain vulnerable when deployed in wireless environments. This underscores the need for a unified approach to improving robustness against both semantic and channel-level uncertainties.


The main contribution of this paper is a novel semantic communication framework, called \textit{Wasserstein distributionally robust wireless semantic communication} (\OurAlg). The proposed framework is designed to enhance robustness against both semantic-level and transmission-level uncertainties, and to generalize across diverse tasks and dynamic wireless environments by explicitly accounting for variability in semantic content and wireless channel conditions. Specifically, \OurAlg formulates a bilevel optimization framework grounded in Wasserstein distributionally robust optimization (WDRO)~\cite{kuhn_wasserstein_2019, gao_wasserstein_2020}. The inner-level problem addresses semantic-level noise by optimizing semantic encoding under worst-case input perturbations, while the outer-level problem mitigates channel impairments by learning transmission strategies that are robust to channel variability. This joint modeling of semantic and transmission uncertainties allows the framework to explicitly handle distinct sources of noise in a principled manner. Furthermore, \OurAlg is model-agnostic and can be integrated with a range of large AI model architectures,  supporting its applicability across different semantic communication scenarios. 

In summary, our key contributions include:

\begin{itemize}
\item  We propose \OurAlg, a novel robust, model-agnostic SemCom framework based on WDRO. The framework is formulated as a bilevel problem to jointly address semantic-level and channel-level uncertainties.

\item We develop a novel algorithm to solve the bilevel problem in \OurAlg by leveraging the dual formulations of both the inner and outer problems. This enables tractable training and supports end-to-end optimization under variability in semantic inputs and wireless channel conditions.

\item We establish theoretical generalization bounds for both optimization levels in \OurAlg, characterizing how the learned semantic and channel models perform under worst-case input perturbations and channel variability, with formal robustness guarantees.

\item We conduct extensive experiments on image and text SemCom tasks. \OurAlg matches state-of-the-art performance under clean conditions and demonstrates greater robustness under semantic perturbations and channel degradations, with consistently more stable PSNR, SSIM, and BLEU trends in noisy scenarios.
\end{itemize}

The rest of this paper is structured as follows. Section~\ref{sec:prev_work} provides an overview of the relevant background and prior works that are closely related to our topics of interest. Section~\ref{sec:prob_formulation} presents our proposed framework, including problem formulation and algorithm designs. Numerical results are discussed in Section~\ref{sec:experiment}, followed by the conclusion in Section~\ref{sec:conclusion}.
\section{Background and Related Works}
\label{sec:prev_work}

This section presents the foundational concepts of semantic communication (SemCom), distributionally robust optimizaton and the role of large AI models. We first introduce the core principles and then analyze the existing limitations, thereby establishing the need for a robust SemCom framework as proposed in this paper.

\subsection{\blue{Principles and Challenges of AI-Enabled Wireless Semantic Communication}}

\begin{figure*}[!t]
\centering
\includegraphics[width=0.75\textwidth]{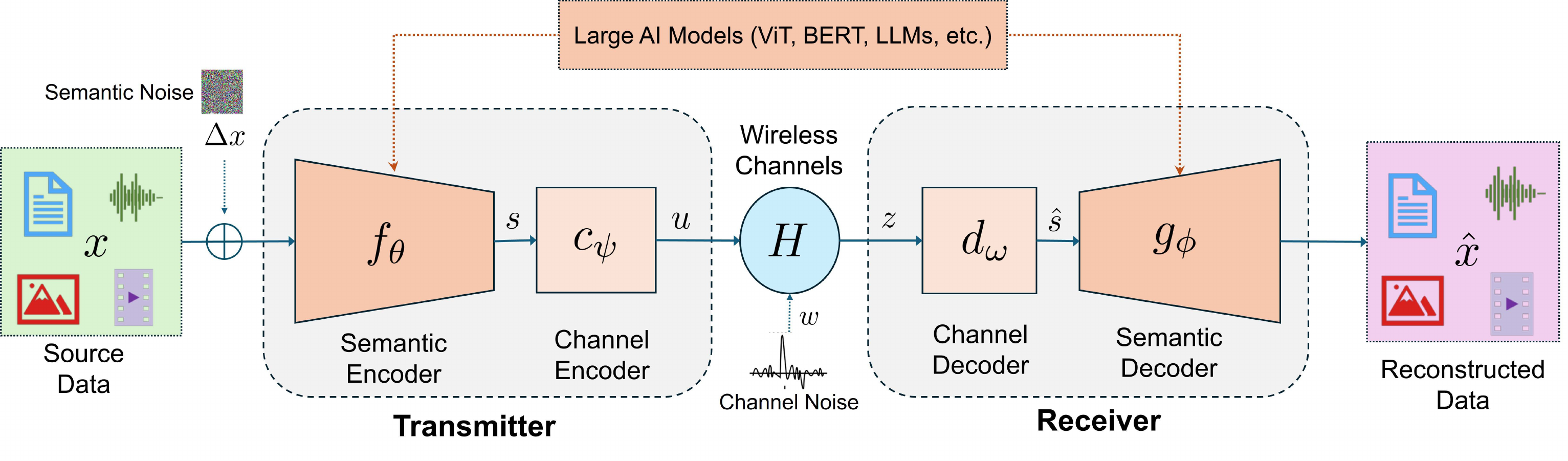}
\caption{Large AI-enabled wireless semantic communication. Source data is encoded into compact, task-relevant representations and transmitted over wireless channels; the receiver reconstructs the intended meaning, even under semantic and channel noise.}
\label{fig:system}
\end{figure*}

\blue{
Wireless semantic communication transmits the meaning of data rather than exact bit sequences~\cite{semcom_principle}. 
Unlike traditional systems that prioritize bit-level fidelity and quality-of-service metrics (e.g., low bit error rate, high signal-to-noise ratio)~\cite{traditional_com}, SemCom aims for fidelity in meaning or task outcome~\cite{semcom_principle, semcom_survey}. 
Recent advances in AI-driven learning have enabled practical realizations of this concept, which allows systems to learn semantic representations and transmit them under adverse channel conditions~\cite{6773024, 8054694}. This capability aligns with Shannon’s vision of semantic-level communication~\cite{shannon} and offers improved resilience in wireless environments. Seminal works have demonstrated these benefits across various modalities. For transmitting text, the work in~\cite{deepsc} introduced end-to-end semantic encoding using deep learning for improving robustness in noisy channels. This work was extended in~\cite{deepsc_st} and~\cite{mem_deepsc} to account for synthesizing audio and to incorporate context-aware question answering. For visual data, autoencoders have been used to transmit images directly over wireless channels~\cite{deepjscc, 10654371}. In the video domain, the solution of~\cite{wireless_deep_video} relied on the use of semantic features and temporal redundancy for efficient video streaming. These works collectively demonstrate that semantic-aware communication enhances efficiency across modalities.

The integration of large-scale AI models, such as Transformers~\cite{transformer} and large language models (LLMs)~\cite{gpt}, has further advanced SemCom systems. The authors in~\cite{deepjscc_text}  employed a pre-trained BERT model~\cite{bert} to enhance text semantic reconstruction in noisy channels. The work in~\cite{semcom_vit} leveraged vision transformer-based models~\cite{vit} to encode semantic features for joint source-channel coding. Similarly, the authors in~\cite{wscs_video} applied large vision-language models~\cite{vit} to video semantic transmission, enabling task-specific feature extraction and cross-modal inference. More recently, GPT-style models are integrated for end-to-end text communication ~\cite{llm_semcom}, highlighting the potential of generative language models in preserving semantic meaning across variable channel conditions.

Despite their promising performance, the existing SemCom methods are generally designed and optimized under nominal conditions and do not explicitly account for robustness to variations in semantic inputs or wireless channel conditions. By operating on high-level, abstract representations, these systems become highly sensitive to two distinct and often coupled sources of uncertainty: \textit{semantic-level noise}, arising from ambiguity or perturbations in the source data, and \textit{transmission-level noise}, encompassing distortions from the physical wireless channel. While seminal AI-based SemCom methods have shown improved resilience over traditional methods, they are typically optimized for average-case performance and lack a formal framework for handling these dual, worst-case uncertainties.}

\subsection{Robustness and Generalization in Wireless SemCom}

Several robustness-oriented strategies have been proposed, typically tailored to specific modalities or noise types.
For semantic-level noise, various methods introduce customized architectures to mitigate input perturbations. In image transmission, masked vector quantization~\cite{masked_vqae} and multi-scale semantic extraction~\cite{robust_sc_2} enhance robustness by improving the semantic representation of visual features. In text-based SemCom, a semantic corrector with non-autoregressive decoding~\cite{robust_sc_4} was used to address categorized semantic impairments. For speech, the framework in~\cite{robust_sc_3} integrated a GAN-based compensator and a semantic probe to preserve intelligibility under semantic distortions. Additionally,~\cite{walid_1} proposed a neuro-symbolic approach that combines signaling games and causal reasoning for context-aware, semantically reliable communication using minimal bits. Regarding channel-level noise, recent works have explored adaptive encoding strategies to improve robustness under time-varying or degraded channel conditions. Examples include transfer learning-based noise estimation~\cite{robust_channel_2} and feedback-aware encoding schemes~\cite{robust_channel}, both of which enhance reconstruction quality. Other efforts focus on improving generalizability and bandwidth efficiency while preserving semantic reliability across dynamic environments~\cite{walid_2}. Although these methods contribute to improved robustness, they are often modality-specific, architecture-dependent, or focused on a single type of noise. 

\blue{
While the works in~\cite{masked_vqae, robust_sc_2, robust_sc_3, robust_sc_4, walid_1, robust_channel_2, robust_channel, walid_2} show the potential of deploying AI for SemCom, they are constrained by three key limitations that motivate our research. First, these methods are predominantly modality-specific, designed for either text, images, or speech, which hinders their generalizability across different communication tasks. Second, they rely on fixed semantic encoding strategies and do not explicitly account for variability or uncertainty in either semantic inputs or channel conditions. Third, these works do not provide formal guarantees on robustness or generalization, especially under worst-case scenarios. These limitations motivate the development of a unified, model-agnostic framework that leverages large AI models while systematically addressing both semantic and transmission-level uncertainties.
}

\blue{
\subsection{Wasserstein Distributionally Robust Optimization}
\label{sec:wdro}

Distributionally Robust Optimization (DRO) is a paradigm designed to handle data uncertainty by training AI models against a ``worst-case" distribution within a predefined ambiguity set, which is distinct from the standard empirical average of the training data~\cite{dro_2016, dro_2018}. This approach yields models that are more resilient to the distributional shifts common in dynamic wireless environments, such as those caused by user mobility, channel fading, or adversarial interference~\cite{sinha_certifying_2020}. While various metrics can be used to define this ambiguity set~\cite{dro_2016, dro_2018}, a particularly powerful variant is Wasserstein DRO (WDRO), which uses the \textit{Wasserstein distance}~\cite{wassertein, kuhn_wasserstein_2019}. Such a metric that quantifies the minimal cost of transporting one probability distribution to another, defined as follows.
\begin{definition}
The \(p\)-\emph{Wasserstein distance}, which measures the cost of transporting probability mass between distributions \(P\) and \(Q\), is defined as:
\begin{align}
    W_p(P, Q) = \inf_{\pi \in \Pi(P, Q)} \left(\mathbf{E}_{(Z, Z') \sim \pi}\left[d^p(Z, Z')\right]\right)^{1/p}.
\end{align}
Here, \(\Pi(P, Q)\) represents the set of all possible joint distributions \(\pi\) with marginals \(P\) and \(Q\), respectively. The random variables $Z \sim P$ and $Z' \sim Q$ represent samples drawn from the respective distributions under the coupling $\pi$, and $d$ is a predefined ground distance metric. The Wasserstein ball, $\mathcal{B}_p(P, \rho) \defeq \{Q: W_p(P, Q) \leq \rho\}$, defines the set of all distributions $Q$ within a $p$-Wasserstein distance $\rho$ from $P$. 
\end{definition}

This metric provides a \textit{geometry-aware} approach to modeling distributional variability, which is essential for the high-dimensional, continuous feature spaces used in semantic communication where the distance between representations is meaningful. Unlike other divergence measures that may not account for the underlying structure of the data space, this geometric sensitivity allows WDRO to model realistic semantic perturbations more effectively. Furthermore, WDRO is designed to build its ambiguity set around the \textit{empirical distribution} of training data~\cite{kuhn_wasserstein_2019}. This makes it effective for practical, data-driven applications, as it can gracefully handle the discrete nature of training samples. The WDRO objective also benefits from a tractable dual reformulation, which enables efficient gradient-based optimization even for deep neural networks~\cite{azizian2023exact}.

}


\section{Wasserstein Distributionally Robust Wireless Semantic Communication}
\label{sec:prob_formulation}

In this section, we first describe a general system model, highlighting key challenges related to semantic and channel noise in wireless SemCom. We then introduce the proposed generalized robust framework designed to tackle the identified challenges. We then elaborate on the algorithm design that emphasizes our proposed method.
\vspace{-5pt}

\subsection{\blue{System Model}}
\begin{figure}[t]
    \centering
    \includegraphics[width=0.45\linewidth]{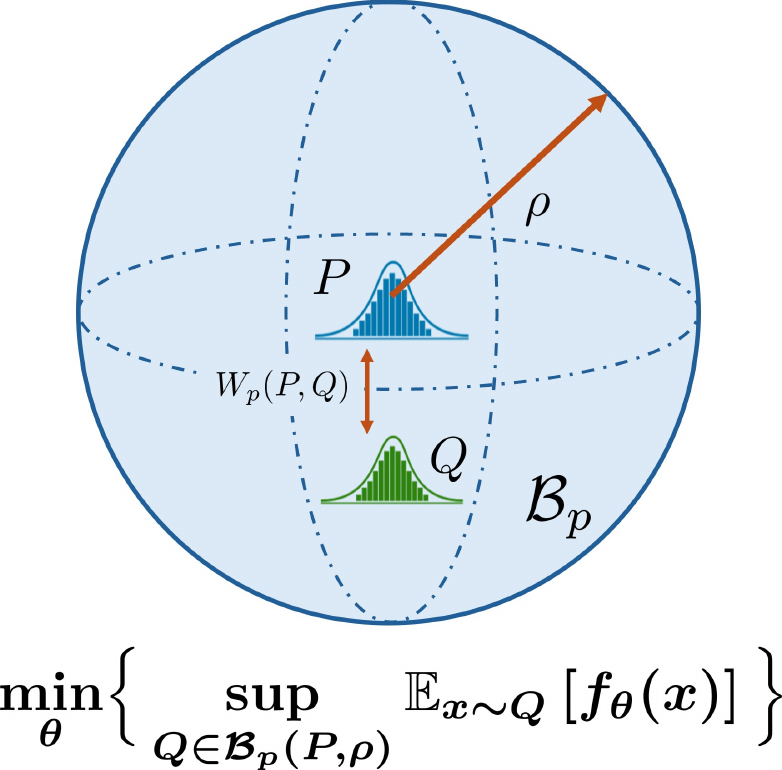}
    \caption{WDRO aims to find model parameters $\theta$ that minimize the worst-case expected objective $f_\theta$ assuming the true data distribution $Q$ is within a small Wasserstein ball $\mathcal{B}_p$ of the empirical distribution $P$.}
    \label{fig:wdro}
\end{figure}

\begin{figure*}[t]
    \centering
    \includegraphics[width=0.7\linewidth]{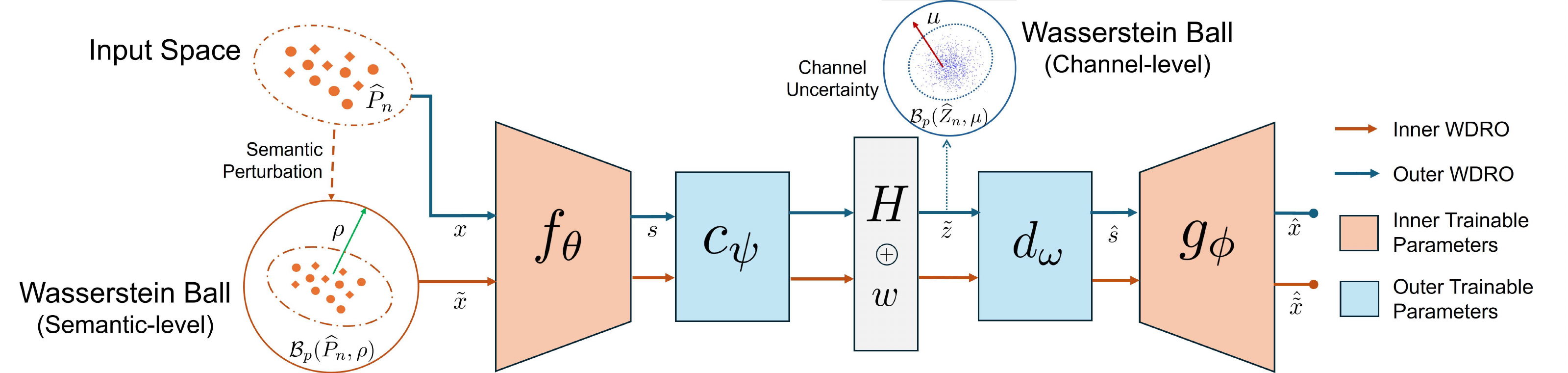}
    \caption{Overview of the \OurAlg framework. The proposed bilevel WDRO model jointly optimizes semantic and channel encoder-decoder pairs for robustness. The inner level (orange) handles semantic input shifts, while the outer level (blue) addresses channel noise.
    }
    \label{fig:bilevel_wdro}
    \vspace{-5pt}
\end{figure*}


We consider an AI-enabled wireless SemCom system (WSCS), depicted in \Cref{fig:system}, designed to transmit the essential meaning embedded within multimodal data such as text, audio, images, or video, instead of raw data bits. The WSCS architecture typically consists of four primary components: a semantic encoder \(f_{\theta}\), a channel encoder \(c_{\psi}\), a channel decoder \(d_{\omega}\), and a semantic decoder \(g_{\phi}\), each parameterized by \(\theta\), \(\psi\), \(\omega\), and \(\phi\), respectively.

Initially, the semantic encoder \(f_{\theta}: \mathcal{X} \to \mathcal{S}\) processes the input data \(x \in \mathcal{X}\), converting it into a concise \emph{semantic representation} \(s = f_{\theta}(x)\) that captures its underlying meaning. This semantic vector \(s \in \mathcal{S}\) is crucial as it contains the distilled essence of the input, optimized for comprehension rather than for bit accuracy. The channel encoder \(c_{\psi}: \mathcal{S} \to \mathcal{U}\) then takes this semantic vector and encodes it into a robust transmittable signal \(u = c_{\psi}(s)\), specifically formatted to withstand the physical limitations and noise characteristics of the wireless channel \(H\). Upon transmission, the signal undergoes various distortions due to noise, fading, or interference, resulting in a corrupted signal \(z \in \mathcal{Z}\) received by the channel decoder~\cite{tse2005fundamentals}. The channel decoder \(d_{\omega}: \mathcal{Z} \to \mathcal{S}\) is responsible for reconstructing the semantic vector \(\hat{s} = d_{\omega}(z)\) from this corrupted signal, effectively filtering out the distortions introduced by the channel. Finally, the semantic decoder \(g_{\phi}: \mathcal{S} \to \mathcal{X}\) takes over to extract and reconstruct the final semantic content \(\hat{x} = g_{\phi}(\hat{s})\), ensuring that the transmitted meaning is accurately recovered.  

In deep learning-based SemCom systems, deep neural networks are typically used for both the encoding and decoding processes to extract, transmit, and reconstruct semantic information. These models used empirical risk minimization (ERM), which minimizes a loss function \(L\) that quantifies the discrepancy between the original input \(x\) and the reconstructed output \(\hat{x}\)~\cite{deepsc,deepjscc}. The learning objective is to minimize the expected loss over the distribution of input data and channel conditions:
\begin{equation}
\min_{\theta, \psi, \omega, \phi} \mathbb{E}_{x \sim \mathcal{X}} [L(x; \theta, \psi, \omega, \phi)],
\end{equation}
where the loss function \(L\) is selected based on the reconstruction objective of the task. We primarily adopt mean squared error (MSE), as it provides a simple yet effective measure of semantic distortion in continuous feature spaces, such as images or latent representations. For the semantic encoder $\theta$ and decoder $\phi$, we can incorporate large AI models like those based on Vision Transformers (ViT)~~\cite{vit} for visual inputs and BERT~~\cite{bert} for textual data due to their proven ability to extract high-level semantic features. This design choice enables our system to generalize across modalities while maintaining semantic fidelity under varying input conditions.

\blue{\textbf{The Challenge of Dual Noise Sources:}}
Conventional approaches to wireless SemCom often adopt the joint source-channel coding (JSCC) paradigm, in which semantic and channel components are optimized under a unified ERM objective~\cite{deepjscc, deepjscc_text}. Despite demonstrating satisfactory performance in controlled settings, such as fixed channel models with stationary noise or limited variability, JSCC solutions lack an explicit separation between semantic representation learning and channel adaptation. 
In particular, the semantic encoder is trained jointly with the channel encoder and decoder to minimize reconstruction loss, thereby encoding not only task-relevant information but also channel-specific statistical features present during training. This might reduce their generality and effectiveness
when deployed in dynamic environments with varying signal-to-noise ratio (SNR) or fading behaviors.
Consequently, the system may not able to maintain semantic fidelity and task relevance across diverse wireless scenarios. This highlights the need for a decoupled design that can independently optimize for semantic expressiveness and channel resilience.

Furthermore, the performance of wireless SemCom systems is fundamentally affected by two types of noise: channel noise and semantic noise. \textit{Channel noise} arises from physical-layer impairments such as thermal noise, fading, and interference~\cite{tse2005fundamentals}, which distort the transmitted signal and may persist despite conventional error correction, particularly in low-SNR or time-varying environments. In contrast, \textit{semantic noise }refers to distortions in meaning that occur even when the signal is correctly decoded, often due to the semantic encoder’s sensitivity to input perturbations or distributional shifts. Adversarial examples or out-of-distribution data~\cite{goodfellow_explaining_2015} can cause the encoder to generate unstable representations that fail to preserve the intended semantics. Notably, this type of degradation arises at the semantic level, independent of the physical channel quality.

While existing SemCom methods have made progress in enhancing robustness, they often focus on mitigating either channel noise or semantic noise in isolation, and are typically designed for specific data modalities such as text~\cite{deepjscc_text} or images~\cite{masked_vqae}. This limits their applicability in more general settings involving diverse input types and jointly occurring distortions. These challenges underscore the need for a unified approach that can jointly handle both channel- and semantic-level uncertainties in a modality-agnostic manner. 
\subsection{\blue{\OurAlg Framework}}
\blue{To address the aforementioned challenges in modern SemCom systems, we propose a novel optimization framework, namely Wasserstein distributionally robust wireless SemCom (\OurAlg). 
Grounded in the principles of WDRO~\cite{kuhn_wasserstein_2019} discussed in~\Cref{sec:wdro}, our approach is designed to systematically and robustly address both semantic and channel-level noise through a novel bilevel formulation.} Unlike conventional ERM, which assumes the training distribution accurately reflects deployment conditions, WDRO explicitly accounts for distributional uncertainty. As illustrated in \Cref{fig:wdro},  WDRO minimizes the worst-case expected loss over a set of distributions within a bounded Wasserstein distance from the empirical distribution. This enables the learned model to remain robust against a wide range of real-world uncertainties, including shifts in semantic content, out-of-distribution inputs, and unpredictable channel conditions, that are common in wireless SemCom systems. 

Building on the foundation of WDRO, we formulate \OurAlg as a bi-level WDRO framework that jointly optimizes the semantic and channel encoding-decoding processes, with each level addressing a distinct source of uncertainty. The \textit{inner level} focuses on robustness to semantic variability in the input representations, while the \textit{outer level} accounts for stochastic perturbations introduced by the wireless transmission channel.  By decoupling and optimizing these two components, the proposed framework improves end-to-end robustness, and thus ensuring high semantic fidelity and reliable communication under heterogeneous and time-varying network conditions.

As illustrated in~\Cref{fig:bilevel_wdro}, the system is trained on \( n \) i.i.d. samples \( \{x_i\}_{i=1}^n \) drawn from an unknown true distribution \( P \), which is approximated by its empirical counterpart \( \widehat{P}_n \). Each input \( x_i \) is mapped to a semantic representation \( s_i = f_\theta(x_i) \) via the semantic encoder \( f_\theta \), 
where \( \delta_{f_\theta(x_i)} \) is the Dirac delta measure centered at the encoded sample. The semantic embedding \( s_i \) is passed through a channel encoder \( c_\psi \), producing a signal that is transmitted through a stochastic channel. The channel introduces distortion via the transformation \( z_i = h \, c_\psi(s_i) + w \), where $h$ is a realization of a random channel state \( H \sim Q_0 \) drawn from a nominal distribution \( Q_0 \), and \( w \) is additive noise such as additive white Gaussian noise (AWGN) or Rayleigh fading. The received signal \( z \) has the empirical distribution over these encoders and channel noise as follows:
\[
\widehat{Z}_n := \frac{1}{n} \sum\nolimits_{i=1}^n \delta_{z_i}, \quad \text{where } z_i = h \cdot c_\psi(f_\theta(x_i)) + w.  
\]
The received signal \( z \) is then decoded by the channel decoder \( d_\omega \) to obtain \( \hat{s} = d_\omega(z) \), which is further mapped back to the semantic space via the decoder \( g_\phi \) to reconstruct the original input as \( \hat{x} = g_\phi(\hat{s}) \).

To formally model uncertainties in both the semantic input and channel transmission process, we define two separate \emph{Wasserstein ambiguity sets}. The first set, \( \mathcal{B}_p(\widehat{P}_n, \rho) \), captures potential semantic-level distributional shifts around the empirical distribution \( \widehat{P}_n \) within radius \( \rho \). The second set, \( \mathcal{B}_p(\widehat{Z}_n, \mu) \), accounts for uncertainties in the distribution of received channel signals \( z \), centered around the nominal distribution \( \widehat{Z}_n \) induced by the channel model under \( Q_0 \) and noise \( w \), with radius \( \mu \). These sets enable a principled treatment of distributional robustness at both semantic and physical layers of the communication system. Based on this setup, the overall objectives for the bi-level problem are:
\begin{align}
\text{\textsc{inner:}} \quad 
& \min_{\theta, \phi} \sup_{Q \in \mathcal{B}_p(\hat{P}_n, \rho)} 
    \mathbb{E}_{x \sim Q}\Big[\ell_{\text{s}}(x, \hat{x}) \mid \psi, \omega\Big] 
    \label{eq:inner} \\
& \qquad~\quad \text{s.t.} \quad 
    \hat{x} = g_\phi(d_\omega(z)) \nonumber\\
 & \qquad~\quad\phantom{\text{s.t.} \quad}  z = h \, c_\psi(f_\theta(x)) + w  \nonumber \\[1ex]
\text{\textsc{outer:}} \quad 
& \min_{\psi, \omega} \sup_{Z \in \mathcal{B}_p(\hat{Z}_n, \mu)} 
    \mathbb{E}_{z \sim Z}\Big[\ell_{\text{c}}(s, \hat{s}) \mid \theta^*\Big]
    \label{eq:outer} \\
& \qquad~\quad \text{s.t.} \quad 
    \hat{s} = d_\omega(z) 
    \nonumber \\
& \qquad~\quad\phantom{\text{s.t.} \quad} 
    s = f_{\theta^*}(x), x \sim \widehat{P}_n,\nonumber
\end{align}

Here $\theta^*$ is the optimal solution obtained from the inner problem. In the inner problem, \(\ell_{\text{s}}(\cdot)\) represents the \textit{semantic reconstruction loss}, measuring the discrepancy between the original input \(x\) and the recovered output \(\hat{x} = g_\phi(d_\omega(z))\), where \(z = h \, c_\psi(f_\theta(x)) + w \). Similarly, in the outer problem,  \(\ell_{\text{c}}(\cdot)\) represents the \textit{channel distortion loss}, which evaluates the distortion between the transmitted semantic representation \(s = f_\theta(x)\) and its recovered version \(\hat{s} = d_\omega(z)\).

\textbf{Inner Level -- Robust Semantic Encoding and Decoding:}
The inner-level objective~\Cref{eq:inner} addresses uncertainty stemming from semantic noise, including misinterpretations, ambiguity, adversarial perturbations, and distributional shifts in the input space (given by $\Delta x$ in \Cref{fig:system}). These challenges are modeled through the Wasserstein ambiguity set $\mathcal{B}_p(\widehat{P}_n, \rho)$, which captures possible semantic variations around the empirical input distribution. The semantic encoder $f_\theta$ and decoder $g_\phi$ are trained to minimize the worst-case semantic reconstruction loss within this uncertainty set, thereby enhancing robustness to input-level perturbations.

A key capability of the inner-level formulation in \OurAlg is its model-agnostic nature. It imposes no constraints on the choice of model architecture or data modality, enabling broad applicability across different communication scenarios. Depending on the task, the semantic encoder-decoder pair can be instantiated using transformer-based models such as BERT~\cite{bert} for textual data, ViT~\cite{vit} for visual inputs, wav2vec~\cite{wav2vec} for audio signals, or multimodal encoders for composite inputs. The semantic loss function $\ell_{\text{s}}(\cdot)$ can also be flexibly chosen to match the modality and task—for example, cross-entropy for classification, mean squared error (MSE) for reconstruction tasks, or perceptual similarity measures for vision or audio applications.

\textbf{Outer Level --  Robust Channel Encoding and Decoding:}
The outer level \Cref{eq:outer} targets physical-layer uncertainties such as channel fading, interference, and signal distortions. It optimizes the channel encoder \(c_\psi\) and decoder \(d_\omega\) to mitigate transmission noise, by minimizing the worst-case distortion under perturbations in the received signal distribution \(\mathcal{B}_p(\widehat{Z}_n, \mu)\). This level may employ either conventional channel coding techniques or deep neural layers trained to be robust under stochastic channel conditions. 
MSE is also a common choice for the channel loss \(\ell_{\text{c}}(\cdot,\cdot) \) when the semantic representation is continuous.

\blue{
It is worth noting that $\rho$ and $\mu$ are independent hyperparameters that operate in distinct spaces, the semantic input space and the channel output space, respectively, and thus cannot be directly compared or jointly optimized through simple scaling, as they govern robustness against different sources of uncertainty.}

\blue{
Together, the bi-level WDRO formulation in \OurAlg systematically addresses distributional uncertainties at both the semantic and channel levels.
This formulation does not treat the two noise sources as independent; rather, it models their \textit{hierarchical dependency}. The inner-level optimization for semantic robustness is rendered channel-aware, as the reconstruction loss is a function of the entire communication chain, thereby ensuring that the learned semantic representations are inherently resilient to distortions introduced by the channel. Symmetrically, the outer-level optimization for channel robustness is semantics-aware, as it is conditioned on the semantic representations derived from the inner loop, ensuring the channel coding is specifically tailored to protect the features deemed most meaningful. This structured methodology renders the complex problem of joint robustness computationally tractable, enhances model generalization by decoupling the primary robustness objectives , and affords practical control over the system's behavior via the independent radii $\rho$ and $\mu$.
}

However, this robustness comes with a tradeoff: optimizing for worst-case scenarios may lead to a more conservative model, potentially sacrificing performance under average or benign conditions. In the context of wireless SemCom, this trade-off is often acceptable since the cost of semantic distortion or transmission failure in rare but adverse conditions can be significantly more detrimental than minor losses in optimal scenarios. To manage this balance, \OurAlg includes a tunable parameter via the radii $\rho$ and $\mu$ of the Wasserstein balls, which serve as regularization parameters controlling the level of robustness. Smaller radii yield solutions closer to standard ERM, favoring average-case performance, while larger radii emphasize robustness to distributional shifts. This formulation enhances resilience to distributional variability and heterogeneity in both semantic inputs and channel conditions, without relying on modality-specific assumptions or post hoc correction mechanisms.

\subsection{\OurAlg: Algorithm Design}

One of the key advantages of WDRO lies in its favorable theoretical and computational properties. In particular, WDRO enjoys strong duality under mild conditions~\cite{kuhn_wasserstein_2019}, which allows the original min-max problem -- defined over an infinite set of probability distributions -- to be reformulated as a finite-dimensional dual problem. This reformulation enables efficient optimization using advanced gradient-based techniques while preserving robustness guarantees.

\subsubsection{Dual Formulation}

To leverage these properties in our bi-level framework, we adopt a dual reformulation approach derived from optimal transport theory~\cite{santambrogio_optimal_2015}, which transforms the original constrained WDRO problem into a more tractable saddle-point optimization problem. In particular, strong duality allows the primal WDRO problem, which involves a supremum over an infinite set of distributions within a Wasserstein ball, to be equivalently expressed as a minimization over a scalar dual variable. 

Concretely, we consider a robust objective defined as:
\[
\sup_{Q \in \mathcal{B}_p(P, \rho)} \mathbb{E}_{Q}[\ell(\cdot)],
\]
where \(\ell(\cdot)\) represents the loss function evaluated under the distribution \(Q\), and \(\rho\) represents the Wasserstein radius, which determines the size of the ambiguity set. 

Under assumptions such as Lipschitz continuity of $\ell(\cdot)$ and compactness of the input space, this problem admits the following dual representation based on Kantorovich duality~\cite{sinha_certifying_2020, gao_distributionally_2016}.
Specially, it can be transformed into a dual form by introducing a Lagrange multiplier \(\lambda \in \mathbb{R}_{+}\) to enforce the Wasserstein constraint \(W_p(P, Q) \leq \rho\), leading to a tractable penalized formulation~\cite{sinha_certifying_2020}:
\begin{align}
\sup_{Q \in \mathcal{B}_p(P, \rho)} &\mathbb{E}_{Q}[\ell(\cdot)]= \nonumber \\ &\inf_{\lambda \geq 0} \left\{ \lambda \rho + \mathbb{E}_{x \sim P} \left[ \sup\nolimits_{\xi} \big( \ell(\xi) - \lambda c(\xi, x) \big) \right] \right\}.
\label{E:duality_Wass}
\end{align}

Here, \(\lambda \) is the dual variable that governs the tradeoff between robustness and fidelity to empirical data. \(x \sim P\) are samples from the empirical distribution. \(\xi\) represents an adversarial perturbation in the input space (e.g., a perturbed latent or semantic representation). \(c(\xi, x)\) is the transportation cost function, which quantifies how much the perturbed sample \(\xi\) deviates from the original input \(x\). This is typically instantiated as the squared Euclidean distance: \(c(\xi, x) = \|\xi - x\|^2\).

Even though this is not a classical Lagrangian dual, it is derived from Kantorovich duality in optimal transport~\cite{kanto_duality} and, under mild conditions, enjoys strong duality—yielding an equivalent and tractable reformulation of the original problem in practice. This formulation has several desirable properties. First, it replaces the intractable optimization over probability measures with a scalar optimization over \(\lambda\), and a point-wise supremum over the perturbation variable \(\xi\). Second, it makes the distributional robustness interpretable: the model is trained to minimize the worst-case expected loss under all distributions within a Wasserstein ball of radius \(\rho\). Finally, the structure of this formulation is favorable to stochastic gradient methods, which enables scalable training even in high-dimensional, deep learning-based architectures.

\textbf{Dual Form for Bi-level Problem:} Based on \Cref{E:duality_Wass}, we derive the dual formulations for both the inner (semantic-level) and outer (channel-level) optimization problems. The inner-level problem focuses on mitigating semantic noise, which captures the inherent variability and uncertainty in how the input \(x\) is semantically encoded. The learning goal is to find the semantic encoder \(f_\theta\) and decoder \(g_\phi\) that minimize the worst-case reconstruction loss \(\ell_{\text{s}}(x, \hat{x})\). To model perturbations in the input space, we denote \(\tilde{x}\) a semantically perturbed version of the input data \(x\), such that the perturbation lies within a Wasserstein ball centered at \(x\).  For example, setting \(\tilde{x} \sim \mathcal{N}(x, \sigma^2 I)\) captures stochastic perturbations arising from natural noise sources such as sensor errors, context ambiguity, or paraphrasing. Alternatively, \(\tilde{x} \) can represent an adversarial sample created by adversarial attacks such as FGSM~\cite{goodfellow_explaining_2015} or PGD~\cite{madry_towards_2019}. Using the Wasserstein duality \eqref{E:duality_Wass}, we have an equivalent problem to the inner problem \Cref{eq:inner} as follows:
\begin{align} 
&\text{\textsc{inner-dual:}} \quad  \min_{\theta, \phi} D(\theta, \phi \vert \psi, \omega), \text{ where } D(\theta, \phi \vert \psi, \omega) \defeq \nonumber\\
&\min_{\lambda \geq 0} \lambda \rho + \mathbb{E}_{{x}\sim\widehat{P}_n} \Bigg[ \sup_{\substack{\tilde{x}: z = h \, c_\psi(f_\theta(\tilde{x})) + w\\ \hat{\tilde{x}} = g_\phi(d_\omega(z))}} \{\ell_{\text{s}}(\tilde{x}, \hat{\tilde{x}}) - \lambda c( x , \tilde{x})\} \Bigl\vert \psi, \omega \Bigg].
\label{eq:inner_dual}
\end{align}

Here, \(c(\tilde{x}, x)\) is the cost function measuring the deviation between the perturbed and original inputs. The dual variable \(\lambda\) controls the balance between robustness to semantic perturbations and fidelity to the observed training data.

Similarly, at the outer level, we aim to learn channel encoding and decoding parameters \( \psi \) and \( \omega \) that minimize the worst-case expected channel loss \( \ell_{\text{c}}(s, \hat{s}) \) over all perturbations \( \tilde{z} \) of the received signal \( z \), subject to a transportation cost constraint. The perturbed signal \( \tilde{z} \) leads to a potentially different reconstructed semantic representation \( \tilde{s} = d_{\omega}(\tilde{z}) \). This accounts not only for physical noise but also for worst-case variations in the channel output. The  objective can be expressed in its dual form as follows:
\begin{align}
&\text{\textsc{outer-dual:}} \quad  \min_{\psi, \omega} D (\psi, \omega \vert \theta^*), \text{ where } D(\psi, \omega \vert \theta^*) \defeq \nonumber \\
&\min_{\gamma \geq 0} \gamma \mu +  \mathbb{E}_{z \sim \widehat{Z}_n} \Bigg[ \sup_{\substack{\tilde{z}: \hat{s} = d_\omega(\tilde{z}), \\ s = f_{\theta^*}(x),\,  x \sim \widehat{P}_n}} \left\{ \ell_{\text{c}}(s, \hat{s}) - \gamma \, c(z, \tilde{z}) \right\} \Bigl\vert \theta^* \Bigg]. 
\label{eq:outer_dual}
\end{align}

Here, \(c(\tilde{z}, z)\) measures the cost of perturbing the transmitted signal. The dual variable \(\gamma\) balances robustness against channel-level uncertainty and adherence to the nominal distribution. By introducing a dual variable associated with the Wasserstein constraint, we effectively decouple the adversarial distributional shift from the primary objective to make the problem analytically and computationally tractable.

\subsubsection{Smooth approximation with Log-sum-exp function}

Although the Wasserstein-based dual forms offer a tractable approach to robust optimization, the hard supremum term \( \mathbb{E}_{x \sim P}\sup_{\xi} \big( \ell(\xi) - \lambda c(\xi, x) \big) \) in \eqref{E:duality_Wass} (and thus in (\textsc{inner-dual}~\Cref{eq:inner_dual} and \textsc{outer-dual}~\Cref{eq:outer_dual}) remains non-smooth and costly to compute, particularly in deep, high-dimensional settings. To overcome this challenge, we replace the inner maximization with a smooth log-sum-exp approximation over a perturbation distribution~\cite{azizian2023exact}:
\begin{equation}
\begin{aligned}
 \epsilon \, \mathbb{E}_{x \sim P}  \log \mathbb{E}_{\xi \sim \tilde{P}(x)}\left[ \exp \left( \frac{ \ell(\xi) - \lambda c(\xi, x)}{ \epsilon } \right) \right]
\end{aligned}
\label{eq:inner_entropic}
\end{equation}
where \(\tilde{P}(x)\) represents a distribution over perturbations of \(x\). This perturbation distribution can be instantiated in multiple ways depending on the robustness modality. Smaller values of \(\epsilon\) result in a tighter approximation to the hard supremum (closer to the original dual), while larger values lead to a smoother landscape that enhances gradient-based learning. The log-sum-exp smoothing transforms the original non-smooth objective into a differentiable form, facilitating efficient stochastic optimization. Importantly, this approximation admits a provable upper bound on the original supremum, with the gap controlled explicitly by $\epsilon$, ensuring robustness is not arbitrarily sacrificed. It also enables scalable training by improving gradient flow in high-dimensional settings~\cite{azizian2023exact}. 

By integrating these techniques, \OurAlg addresses both the theoretical complexities inherent in WDRO and the practical challenges in implementing these models in real-world wireless SemCom systems.

\subsubsection{Bi-Level Optimization Procedure}

\begin{algorithm}[t]
\caption{\blue{Training Algorithm for \OurAlg}}
\label{alg:wasecom}
\begin{algorithmic}[1]
\REQUIRE Training data $\{x_i\}_{i=1}^n \sim \widehat{P}_n$, channel model $h$, noise model $n$; 
Wasserstein radii $\rho$, $\mu$; 
no. of training steps $T$
\STATE Initialize model parameters $\theta^{(0)}$, $\phi^{(0)}$, $\psi^{(0)}$, $\omega^{(0)}$ and dual variables $\lambda \geq 0$, $\gamma \geq 0$, $\epsilon \geq 0$
\FOR{$t = 1$ to $T$}
    \STATE Sample minibatch $\{x_i\}$ from $\widehat{P}_n$ \label{alg:sampling}

    \FOR{\textsc{outer loop}} \label{alg:outer_start}
        \STATE Compute semantic representation $s_i \gets f_{\theta^{(t-1)}}(x_i)$ \label{alg:outer_encode_sem}
        \STATE Transmit through channel: $z_i \gets h(c_{\psi^{(t-1)}}(s_i)) + w$ \label{alg:outer_signal}
        \STATE Generate perturbed signals $\tilde{z}_i \in \mathcal{B}_p(z_i, \mu)$
        \label{alg:outer_pertubed_signal}
        \STATE Decode received signal: $\hat{s}_i \gets d_{\omega^{(t-1)}}(\tilde{z}_i)$
        \label{alg:outer_decode_sem}
         \STATE Update $\psi^{(t)} \text{ and } \omega^{(t)}$ by solving problem~\Cref{eq:outer_dual} using gradient-based methods. \label{alg:outer_opt}
    
    \FOR{\textsc{inner loop}} \label{alg:inner_start}
        \STATE Generate semantic perturbation $\tilde{x}_i \in \mathcal{B}_p(x_i, \rho)$ \label{alg:inner_input}
        \STATE Compute encoded semantic $\tilde{s}_i \gets f_{\theta^{(t-1)}}(\tilde{x}_i)$ \label{alg:inner_encode_sem}
        \STATE Compute channel output $z'_i \gets h(c_{\psi^{(t)}}(\tilde{s}_i)) + w$
        \label{alg:inner_signal}
        \STATE Decode and reconstruct: $\hat{\tilde{x}}_i \gets g_{\phi^{(t-1)}}(d_{\omega^{(t)}}(z'_i))$ \label{alg:inner_decode_sem}
        \STATE Update $\theta^{(t)} \text{ and }\phi^{(t)}$ by solving problem~\Cref{eq:inner_dual} using gradient-based methods.
        \label{alg:inner_opt}
    \ENDFOR \label{alg:inner_end}
    \ENDFOR  \label{alg:outer_stop}

\ENDFOR
\end{algorithmic}
\end{algorithm}

\blue{To implement the proposed bi-level WDRO framework, we design an iterative training algorithm that alternates between solving the \textit{inner semantic-level} and \textit{outer channel-level} dual problems using gradient-based updates in~\Cref{alg:wasecom}. This training procedure is designed to instill resilience into the final, fixed model, preparing it for operational deployment. The algorithm employs an alternating optimization scheme, a standard approach for bilevel problems, where two sets of parameters are refined in an alternating fashion within each training iteration~\cite{bilevel_alternative}. The goal is to progressively refine the encoder and decoder parameters to improve robustness against both semantic uncertainty (e.g., paraphrasing or sensory noise) and channel-induced distortions (e.g., transmission noise or channel fading).}


In each training iteration, a mini-batch of samples $\{x_i\}$ is drawn from the empirical distribution $\widehat{P}_n$ (line~\ref{alg:sampling}). The outer loop (lines~\ref{alg:outer_start}--\ref{alg:outer_stop}) aims to enhance robustness against channel-level noise and shifts, leveraging the optimized semantic encoder from the inner level. For each $x_i$, the semantic encoder $f_{\theta^{(t-1)}}$, using parameters from the previous iteration, computes the latent representation $s_i$ (line~\ref{alg:outer_encode_sem}). The channel encoder $c_{\psi^{(t-1)}}$ maps this to a signal $u_i$, which is perturbed by the channel and noise to produce $z_i$ (line~\ref{alg:outer_signal}). 
The perturbation $\tilde{z}_i \in \mathcal{B}_p(z_i, \mu)$ simulates the worst-case channel effect within a Wasserstein ball of radius $\mu$ (line~\ref{alg:outer_pertubed_signal}). 
The channel decoder $d_{\omega^{(t-1)}}$ reconstructs the semantic signal $\hat{s}_i$ (line~\ref{alg:outer_decode_sem}). $\psi^{(t)}$ and $\omega^{(t)}$ are updated by solving the outer-level dual problem in Eq.~\Cref{eq:outer_dual}, minimizing the worst-case channel distortion (line~\ref{alg:outer_opt}). 

Following this, the inner loop (lines~\ref{alg:inner_start}--\ref{alg:inner_end}) addresses robustness to semantic perturbations by solving the semantic-level dual problem (cf. \Cref{eq:inner_dual} or its entropic variant in \Cref{eq:inner_entropic}). For each input $x_i$, a perturbed version $\tilde{x}_i \in \mathcal{B}_p(x_i, \rho)$ is generated within a Wasserstein ball of radius $\rho$ (line~\ref{alg:inner_input}). This perturbation simulates semantic ambiguity arising from context shifts or adversarial modifications. The perturbed sample $\tilde{x}_i$ is encoded via the semantic encoder $f_{\theta^{(t-1)}}$ (line~\ref{alg:inner_encode_sem}, transmitted through the current channel encoder $c_{\psi^{(t)}}$, and passed through the channel $h$ with noise $n$ to produce the signal $z_i$ (line~\ref{alg:inner_signal}). The received signal is decoded and reconstructed via the channel decoder $d_{\omega^{(t)}}$ and semantic decoder $g_{\phi^{(t-1)}}$ to obtain $\hat{\tilde{x}}_i$ (line~\ref{alg:inner_decode_sem}). $\theta^{(t)}$ and $\phi^{(t)}$ are then updated to minimize the worst-case semantic loss under this perturbation, by solving the semantic dual formulation (line~\ref{alg:inner_opt}).


\blue{\textbf{Model Deployment and Inference:} Upon completion of the training phase detailed in~\Cref{alg:wasecom}, the optimized model parameters ($\theta, \phi, \psi, \omega$) are utilized for practical deployment, where the system operates as a standard, feed-forward semantic communication pipeline. The process at the transmitter begins with the robust semantic encoder ($f_\theta$) processing a source data sample $x$ to extract a compact and meaningful representation, which is then prepared for transmission by the channel encoder ($c_\psi$). Following propagation over the physical wireless channel, the receiver employs the channel decoder ($d_\omega$) to recover the semantic representation from the incoming signal. Subsequently, the semantic decoder ($g_\phi$) uses this recovered representation to produce the final reconstruction $\hat{x}$. 
}

\section{\OurAlg: Theoretical Analysis} \label{sec:generalization}

In this section, we establish the generalization and robustness guarantees of \OurAlg by deriving uniform convergence bounds for both levels of the proposed bi-level framework. The goal is to establish formal guarantees that the learned semantic and channel models perform reliably not only on the training distribution but also under adversarial and out-of-distribution shifts captured by Wasserstein balls around the empirical distributions. Specifically, we analyze the excess risk at the inner semantic-level and the outer channel-level, leveraging the dual formulations of WDRO.

\subsection{Preliminaries and Assumptions}

We begin by formalizing the notation and assumptions used in the subsequent analysis. For better presentation, let us denote \(\vartheta = (\theta, \phi)\) correspond to the semantic encoder \(f_\theta\) and decoder \(g_\phi\), and the associated semantic reconstruction loss is denoted as \(\ell_{\text{s}}(x; \vartheta)\). Also let \(\varphi = (\psi, \omega)\) parameterize the channel encoder \(c_\psi\) and decoder \(d_\omega\), with a corresponding channel distortion loss \(\ell_{\text{c}}(s, z; \varphi)\), where \(s = f_\theta(x)\) is the semantic representation and \(z\) is the received signal corrupted by the channel.



We adopt the following assumptions, standard in the literature on distributionally robust optimization~\cite{sinha_certifying_2020, kuhn_wasserstein_2019}.

\begin{assumption}[Convexity of Transportation Cost] \label{ass:cost}
The transportation cost $c: \mathcal{X} \times \mathcal{X} \to \mathbb{R}_+$ is continuous, and for all $x_0 \in \mathcal{X}$, the function $c(\cdot, x_0)$ is $1$-strongly convex. A typical instantiation is the squared Euclidean cost: $c(x, x') = \|x - x'\|^2$.
\end{assumption}

\begin{assumption}[Lipschitz Continuity of Losses] \label{ass:lipschitz}
The semantic and channel losses are Lipschitz continuous with respect to their respective input parameters:
\begin{align*}
\text{(a)} \quad & |\ell_{\text{s}}(x_i; \vartheta) - \ell_{\text{s}}(x_j; \vartheta)| \leq L_{\text{s}} \|x_i - x_j\|, \\
\text{(b)} \quad & |\ell_{\text{c}}(s, z_i; \varphi) - \ell_{\text{c}}(s, z_j; \varphi)| \leq L_{\text{c}} \|z_i - z_j\|.
\end{align*}
\end{assumption}

\begin{assumption}[Smoothness of Loss Functions] \label{ass:smooth}
The loss functions $\ell_{\text{s}}$ and $\ell_{\text{c}}$ are smooth, i.e., they have Lipschitz continuous gradients with respect to their respective inputs.
\end{assumption}

\begin{remark}[On Assumptions]
These assumptions ensure the tractability and stability of our optimization framework. The strongly convex transport cost (Assumption~\ref{ass:cost}) enables dual reformulation of WDRO, while Lipschitz continuity and smoothness of the losses (Assumptions~\ref{ass:lipschitz}, \ref{ass:smooth}) support convergence of gradient-based methods. These conditions are typically satisfied in SemCom models using standard neural architectures and common loss functions.
\end{remark}

\subsection{Robust Surrogate Risk and Excess Risk for Bi-level WDRO}
\label{sec:surrogate_excess_bilevel}

Let $\mathcal{F}_s = {\ell_s(\cdot;\vartheta): \vartheta\in\Theta}$ and $\mathcal{F}_c = {\ell_c(s,\cdot;\varphi): \varphi\in\Phi}$ be the sets of loss functions realized by the semantic and channel models (with $s$ treated as a contextual parameter for $\mathcal{F}_c$). 

To support the generalization analysis, we define the sets of loss functions induced by these models. Let \(\mathcal{L}_{\text{s}} :=   \{x \mapsto \ell_{\text{s}}(x; \vartheta): \vartheta\in\Theta\)\} denote the class of semantic loss functions as \(\vartheta\) varies over the parameter space \(\Theta\). Similarly, let \(\mathcal{L}_{\text{c}}:= \{z \mapsto \ell_{\text{c}}(s, z; \varphi):\varphi\in\Phi\}\) denote the class of channel loss functions parameterized by \(\varphi \in \Phi\) with semantic input \(s\) held fixed.

Let \(h_{\text{s}} \in \mathcal{L}_{\text{s}}\) and \(h_{\text{c}} \in \mathcal{L}_{\text{c}}\). We define the following dual surrogate objectives based on the dual formulations of WDRO:
\begin{align}
S_\lambda(x; h_{\text{s}}) := \sup_{\tilde{x} \in \mathcal{X}} \left\{ h_{\text{s}}(\tilde{x}) - \lambda c(\tilde{x}, x) \right\}, \label{eq:semantic_surrogate} \\
C_\gamma(z; h_{\text{c}}) := \sup_{\tilde{z} \in \mathcal{Z}} \left\{ h_{\text{c}}(\tilde{z}) - \gamma c(\tilde{z}, z) \right\}. \label{eq:channel_surrogate}
\end{align}

\begin{definition}[Expected Risks and Surrogate Risks]
Let $P$ be the distribution over inputs \(x \in \mathcal{X}\), and \(Z\) be the distribution over channel outputs \(z \in \mathcal{Z}\). Let $s = f_\theta(x)$ be the semantic representation. Then, the \emph{expected risks} are:
\[
  \mathscr{L}(P, h_{\text{s}}) := \mathbb{E}_{x \sim P} [h_{\text{s}}(x)], \quad
  \mathscr{L}(Z, h_{\text{c}}) := \mathbb{E}_{z \sim Z} [h_{\text{c}}(z)]
\]
The corresponding \emph{robust surrogate risks} are defined as:
\[
\mathscr{L}_\rho^\lambda(P, h_{\text{s}}) := \mathbb{E}_{x \sim P}[S_\lambda(x; h_{\text{s}})] + \lambda \rho^2
\]
\[
\mathscr{L}_\mu^\gamma(Z, h_{\text{c}}) := \mathbb{E}_{z \sim Z}[C_\gamma(z; h_{\text{c}})] + \gamma \mu^2
\]
\end{definition}

\begin{definition}[Excess Risks and Robust Excess Risks]
The excess risks are:
\begin{align}
\mathscr{E}(P, h_{\text{s}}) := \mathscr{L}(P, h_{\text{s}}) - \inf_{h' \in \mathcal{L}_{\text{s}}} \mathscr{L}(P, h'), \\
\mathscr{E}(Z, h_{\text{c}}) := \mathscr{L}(Z, h_{\text{c}}) - \inf_{h' \in \mathcal{L}_{\text{c}}} \mathscr{L}(Z, h')
\end{align}

The corresponding robust excess risks are defined as:
\begin{align}
\mathscr{E}_\rho^\lambda(P, h_{\text{s}}) := \mathscr{L}_\rho^\lambda(P, h_{\text{s}}) - \inf_{h' \in \mathcal{L}_{\text{s}}} \mathscr{L}_\rho^\lambda(P, h'),\\
\mathscr{E}_\mu^\gamma(Z, h_{\text{c}}) := \mathscr{L}_\mu^\gamma(Z, h_{\text{c}}) - \inf_{h' \in \mathcal{L}_{\text{c}}} \mathscr{L}_\mu^\gamma(Z, h')
\end{align}
\end{definition}



We now establish that the surrogate excess risks, defined via the dual formulation of WRDO, can serve as accurate approximations to the worst-case excess risks over Wasserstein balls at both levels of our bilevel optimization problem.

\begin{lemma}[Bi-level Surrogate Excess Risk Bounds]
\label{lem:excess_risk}
Suppose $f \in \mathcal{F}_{\text{s}}$ is $L_{\text{s}}$-Lipschitz and $g \in \mathcal{F}_{\text{c}}$ is $L_{\text{c}}$-Lipschitz. If $\lambda \ge L_{\text{s}}/\rho$ and $\gamma \ge L_{\text{c}}/\mu$, then for any $P' \in \mathcal{B}(P, \rho)$ and $Z' \in \mathcal{B}(Z, \mu)$:
\begin{align}
|\mathscr{E}(P', h_{\text{s}}) - \mathscr{E}_\rho^\lambda(P, h_{\text{s}})| \le 2L_{\text{s}} \rho + |\lambda - \lambda^*| \rho^2,\\
|\mathscr{E}(Z', h_{\text{c}}) - \mathscr{E}_\mu^\gamma(Z, h_{\text{c}})| \le 2L_{\text{c}} \mu + |\gamma - \gamma^*| \mu^2,
\end{align}
where $\lambda^*$ and $\gamma^*$ are the optimal dual variables corresponding to the inner and outer WDRO problems, respectively.
\end{lemma}

We provide the proof of \Cref{lem:excess_risk} in Appendix~\ref{proof:excess_risk}. 

\begin{remark}
	\Cref{lem:excess_risk} demonstrates that the robust excess risks defined via the dual (penalized) objectives are tightly coupled to the actual worst-case risks over the Wasserstein ambiguity sets. The approximation gap consists of two interpretable terms. The first term, \(2L_{\text{s}}\rho\) (or \(2L_{\text{c}}\mu\) for the channel), quantifies the inherent cost of robustness under distributional shift. In wireless SemCom, this reflects the system’s tolerance to semantic or channel-level perturbations, such as adversarial inputs, sensor noise, or fading. Its linear dependence on the Lipschitz constant arises from the bounded variation of the loss under small perturbations, ensuring that the surrogate and worst-case risks are tightly coupled when \(\rho\) or \(\mu\) is small.
    The second term, \(|\lambda - \lambda^*| \rho^2\) (and analogously for $\gamma$), captures the penalty from suboptimal tuning of the dual regularization parameters. Overly conservative or insufficiently protective values can lead to either unnecessary resource usage or degraded robustness. This motivates practical tuning of \(\lambda\) (or \(\gamma\)) based on reliability or task-specific constraints.
    Together, these bounds justify the use of dual surrogate risks in \OurAlg as accurate and efficient proxies for true worst-case performance. They apply to both semantic and channel levels, supporting robust end-to-end communication in uncertain environments. Note that $\mathscr{L}_{\rho}^{\lambda^*} (P,h_s)$ and $\mathscr{L}_{\mu}^{\gamma^*} (Z,h_c)$ are the same as $\gB(P,\rho)$ and $\gB(Z, \rho)$-worst-case risk thanks to the strong duality in \cref{E:duality_Wass}, obtained with $\rho, \mu > 0$.

\end{remark}

We now present a generalization result for the bilevel WDRO framework, which demonstrates that minimizing the empirical surrogate risks at both the semantic and channel levels yields models that generalize well to the population setting under distributional shifts captured by Wasserstein balls.

\begin{theorem}[Robust generalization bounds] \label{thrm:excess_risk} 
Let \(\widehat{h}_{\text{s}} \in \mathcal{L}_{\text{s}}\) and \(\widehat{h}_{\text{c}} \in \mathcal{L}_{\text{c}}\) be \(\varepsilon\)-optimal solutions to the empirical surrogate risk minimization problems for the inner (semantic) and outer (channel) levels, respectively. Consider that \Cref{ass:cost}--\ref{ass:smooth} hold, and the losses are uniformly bounded by \(M\), i.e., \(|\ell_{\text{s}}(x; \vartheta)| \le M\) and \(|\ell_{\text{c}}(s, z; \varphi)| \le M\) for all inputs. Then, with probability at least \(1 - \delta\) over the training sample of size \(n\), the worst-case excess risks are bounded as:
\begin{align*}
\mathscr{E}(Z, \widehat{h}_{\text{s}}) \le 
\frac{48 \mathscr{C}(\mathcal{L}_{\text{s}})}{\sqrt{n}} +
2M \sqrt{ \frac{2 \log(2 / \delta)}{n} } +
\varepsilon + g(\rho, \lambda)\\
\mathscr{E}(Z, \widehat{h}_{\text{c}}) \le 
\frac{48 \mathscr{C}(\mathcal{L}_{\text{c}})}{\sqrt{n}} +
2M \sqrt{ \frac{2 \log(2 / \delta)}{n} } +
\varepsilon + g(\mu, \gamma),
\end{align*}
for any \(P' \in \mathcal{B}_p(P, \rho)\) and \(Z' \in \mathcal{B}_p(Z, \mu)\), where the robustness penalties are defined as:
\[
g(\rho, \lambda) = 2L_{\text{s}} \rho + |\lambda - \lambda^*| \rho^2, \quad
g(\mu, \gamma) = 2L_{\text{c}} \mu + |\gamma - \gamma^*| \mu^2,
\]
and \(\mathscr{C}(\mathcal{L})\) denotes the complexity of the function class \(\mathcal{L}\), measured via the Dudley entropy integral:
\[
\mathscr{C}(\mathcal{L}) := \int_0^\infty \sqrt{ \log \mathcal{N}(\mathcal{L}, \|\cdot\|_\infty, \epsilon) } \, d\epsilon,
\]
where \(\mathcal{N}(\mathcal{L}, \|\cdot\|_\infty, \epsilon)\) is the covering number of \(\mathcal{L}\) with respect to the uniform norm. 
\end{theorem}

The proof can be found in~\Cref{proof:thrm_excess_risk}. 
We sketch the  proof as follows: First, we apply uniform convergence bounds to control the deviation between the empirical surrogate risk and its population counterpart. This involves bounding the empirical Rademacher complexity of the function classes \(\mathcal{L}_{\text{s}}\) and \(\mathcal{L}_{\text{c}}\), which yields the generalization term involving \(\mathscr{C}(\mathcal{L}) / \sqrt{n}\). Second, we invoke the sandwich lemma for surrogate excess risk (previously stated as Lemma~\ref{lem:excess_risk}) to relate the true worst-case excess risk to the surrogate excess risk. This results in an additive penalty term \(g(\cdot, \cdot)\), which quantifies the approximation error due to robustness and the suboptimal choice of dual variables \(\lambda\) and \(\gamma\). As a concrete example, suppose \(\mathcal{L}_{\text{s}}\) corresponds to a class of linear predictors \(\{x \mapsto \langle \theta, x \rangle : \|\theta\|_2 \le C\}\). In this case, the covering number satisfies \(\mathcal{N}(\mathcal{L}_{\text{s}}, \|\cdot\|_\infty, \epsilon) = \left(1 + \frac{2C}{\epsilon}\right)^d\), and the Dudley integral yields a complexity bound \(\mathscr{C}(\mathcal{L}_{\text{s}}) \le \frac{3}{2} C L_{\theta} \sqrt{d}\), where \(L_\theta\) is the Lipschitz constant of the parameterization. This shows that \(\mathscr{C}(\mathcal{L})\) grows at a moderate rate in the dimension and hypothesis class size, making the bound practically meaningful~\cite{lee_minimax_2018}. 


\begin{remark}
\Cref{thrm:excess_risk} provides a generalization guarantee for \OurAlg under distributional uncertainty. It shows that minimizing the empirical surrogate risks at both semantic and channel levels yields models that generalize to unseen data and channel conditions, with convergence rate \(\mathcal{O}(1/\sqrt{n})\), matching classical learning theory. Importantly, the additional terms \(g(\rho, \lambda)\) and \(g(\mu, \gamma)\) quantify the robustness-performance tradeoff in a wireless SemCom context. In SemCom, where incorrect reconstruction of meaning can have a much higher cost than bit errors, these terms explicitly bound how much robustness to semantic distortion and channel degradation impacts performance. When \(\rho\) and \(\mu\) are small, the bounds recover standard ERM behavior, indicating strong performance under nominal conditions. As they increase, the model becomes more resilient to shifts, at the cost of possible conservatism. This interpolation is useful in wireless environments with unpredictable conditions, allowing system behavior to be tuned for specific needs. Thus, the theoretical bounds justify the design of \OurAlg and offer actionable guidance for its deployment in diverse wireless settings.
\end{remark}

\vspace{-15pt}
\blue{
\subsection{Convergence of \OurAlg}

The alternating structure in Algorithm~\ref{alg:wasecom} ensures that the learned SemCom system is jointly robust, which tolerates both semantic-level ambiguity and channel-level signal degradation. Due to the inherent non-convexity of AI models, global optimality cannot be theoretically guaranteed. This limitation applies broadly to modern adversarial training and DRO frameworks~\cite{Kawaguchi2019GradientDF, dinh2017}. However, the objective of our framework is to find a solution that corresponds to a robust stationary point, which represents a locally optimal solution that is stable against worst-case perturbations.

The convergence of \OurAlg to stationary solutions is established through three complementary theoretical principles. \textit{First}, by leveraging Kantorovich duality~\eqref{eq:inner_dual}\eqref{eq:outer_dual} and applying log-sum-exp smoothing~\eqref{eq:inner_entropic}, we transform the original intractable worst-case objective into a differentiable surrogate with Lipschitz-continuous gradients. Under~\Cref{ass:cost}--\ref{ass:smooth}, this reformulation yields a well-behaved optimization landscape amenable to gradient-based methods~\cite{esfahani_data-driven_2017,gao_distributionally_2016}. \textit{Second}, Algorithm~\ref{alg:wasecom} adopts an alternating gradient descent--ascent scheme, a well-established approach for non-convex min--max optimization. Under the smoothness conditions in Assumptions~\ref{ass:lipschitz}--~\ref{ass:smooth}, such alternating methods are proven to converge to first-order stationary points~\cite{pmlr-v119-lin20a,bilevel_alternative}, with iterates $\{{\theta}^{(t)}, {\varphi}^{(t)}, {\psi}^{(t)}, {\omega}^{(t)}\}$ satisfying $\liminf_{t\to\infty} \mathbb{E}[\|\nabla \mathcal{L}({\theta}^{(t)})\|^2] = 0$, indicating convergence to a point where no gradient-based improvement is possible. \textit{Third}, Theorem~\ref{thrm:excess_risk} provides the critical quality guarantee: it formally proves that any solution minimizing our empirical surrogate risk, the robust stationary point found by Algorithm~\ref{alg:wasecom}, achieves bounded worst-case excess risk on unseen data within Wasserstein balls. This ensures that the \textit{local optimum} identified by our algorithm is not arbitrary, but provably robust and generalizable under distributional shifts encountered in wireless environments.

}

\section{Numerical Results}
\label{sec:experiment}

\subsection{Experiment Setup}
\subsubsection{Datasets}
To evaluate the effectiveness of our proposed method for robust wireless semantic communication, we employ two widely used datasets that represent distinct modalities: visual and textual. For the image-based tasks, we use \textit{CIFAR-10}, a benchmark dataset comprising 60,000 color images of size $32 \times 32$, evenly distributed across 10 object categories with 6,000 samples per class. Its diversity and manageable scale make it suitable for assessing visual semantic communication performance. For text-based evaluation, we utilize the \textit{Europarl} corpus, a large-scale multilingual parallel dataset extracted from the proceedings of the European Parliament spanning 1996 to 2011. This corpus includes sentence-aligned text in 21 European languages from various linguistic families, with bilingual pairings ranging from 400,000 to 2 million sentence pairs depending on the language combination. Additionally, it provides monolingual corpora containing 7 million to 54 million words per language for nine languages. Europarl is a well-established benchmark for machine translation and semantic evaluation in multilingual settings.

\subsubsection{Models}
We leverage different large AI models as modality-specific backbones to perform robust semantic encoding and decoding. For image-based communication, we use a Denoising Autoencoder Vision Transformer (DAE-ViT)~\cite{vit} to extract high-level semantic features from input images. The DAE-ViT encoder splits images into patches, embeds them, and processes them through Transformer layers to produce a compact semantic representation. This is passed through a 2-layer MLP-based channel encoder to simulate modulation before being transmitted over a differentiable wireless channel. On the receiver side, an MLP-based channel decoder recovers the representation, which is then reconstructed by a DAE-ViT decoder.
For text-based communication, we use BERT-Base~\cite{bert} as the semantic encoder, which processes tokenized input text and outputs contextual embeddings. These embeddings are mean-pooled into a semantic vector, passed through an MLP-based channel encoder, followed by a channel model and a channel decoder, as in the image case. The output is fed into a Transformer-based decoder, initialized from a pre-trained model and fine-tuned to reconstruct the original sentence.

\blue{\subsubsection{Robustness Simulation Strategy} Our experimental evaluation is designed to quantitatively validate the dual-robustness capabilities of \OurAlg against two distinct types of distributional shifts. We assess semantic robustness by employing FGSM adversarial attacks at varying levels of intensity. Specifically, we evaluate performance under attacks with an $\ell_\infty$-norm perturbation strength configured to represent both moderate (10\%) and severe (30\%) noise conditions. This allows for a precise analysis of performance degradation as the semantic attack becomes more potent. Concurrently, we evaluate channel robustness by measuring performance across a wide spectrum of channel conditions, from 0 to 30 dB Signal-to-Noise Ratio (SNR), under both AWGN and Rayleigh fading models. This wide SNR range simulates environments from very poor to excellent channel quality. This methodology allows for a quantitative analysis of the framework's resilience to both the statistical nature of the channel and its time-varying quality.}

\subsubsection{Baselines}
To benchmark the effectiveness of \OurAlg, we compare it against several state-of-the-art SemCom methods. 
For image transmission, we consider two baselines:
(1) DeepJSCC\cite{deepjscc}, an end-to-end model that jointly optimizes source and channel coding for wireless image transmission;
(2) DeepSC-RI\cite{robust_sc_2}, which improves robustness by incorporating a multi-scale semantic extractor based on ViT and a cross-attention-based semantic fusion module.
For text-based tasks, we include DeepSC~\cite{deepsc}, a Transformer-based model designed to preserve semantic meaning during transmission by optimizing both system capacity and semantic accuracy.

\subsubsection{Evaluation Metrics}
We evaluate model performance using standard metrics for both image and text modalities. For images, we use \textit{Peak Signal-to-Noise Ratio (PSNR)} to measure the fidelity of reconstruction, where higher values indicate better visual quality, and \textit{Structural Similarity Index (SSIM)} to assess perceived structural similarity between original and reconstructed images. For text, we use the \textit{BLEU score}, which compares machine-generated sentences to reference translations based on n-gram overlap, with scores closer to 1 indicating higher semantic similarity and better preservation of meaning.

\subsubsection{Training Details} 
\blue{
We train the \OurAlg pipeline end-to-end using Algorithm~\ref{alg:wasecom} with the Adam optimizer~\cite{Kingma2014AdamAM}, batch size of 128, and 100 training epochs. The outer loop optimizes the channel encoder and decoder under worst-case channel perturbations within Wasserstein ball $\mathcal{B}_p(\hat{\mathcal{Z}}_n, \mu)$ with radius $\mu = 0.01$, simulated via differentiable AWGN or Rayleigh fading models. The inner loop enhances robustness against semantic distributional shifts within $\mathcal{B}_p(\hat{\mathcal{P}}_n, \rho)$ semantic radius $\rho = 0.05$ using FGSM adversarial perturbations. Dual variables $\lambda$ and $\gamma$ are initialized to 1.0 and updated automatically via gradient descent with learning rate $5 \times 10^{-3}$. The log-sum-exp smoothing parameter is set to $\epsilon = 0.1$ for all experiments. All components are differentiable, which enables end-to-end gradient-based optimization. Gradient clipping with maximum norm 1.0 is applied to ensure training stability.}

Experiments are conducted using Python 3.10, PyTorch 2.1.0, and CUDA 12.1 on an Intel® Xeon® W-3335 server with 512GB RAM and NVIDIA RTX 4090 GPUs.

\subsection{Main Results}

\subsubsection{Performance on Image Semantic Communication}

We evaluate the robustness of image-based SemCom under varying levels of semantic and channel noise conditions using the CIFAR-10 dataset. The proposed method, \OurAlg, is benchmarked against DeepJSCC, a non-robust baseline, and DeepSC-RI, which primarily enhances robustness to semantic perturbations. The evaluation spans three semantic conditions (clean input, 10\% FGSM, and 30\% FGSM perturbation) and two channel models: AWGN and Rayleigh fading, across a wide SNR range from 0 to 30~dB. We report performance using PSNR and SSIM metrics.

\blue{As depicted in Fig.~\ref{fig:combined} and~\ref{fig:combined_2}, and detailed quantitatively in Tables~\ref{tab:Rayleigh_results} and~\ref{tab:AWGN_results}, all methods achieve competitive reconstruction quality under clean (noise-free) conditions. Under AWGN with clean inputs (Fig.~\ref{fig:combined}, Table~\ref{tab:AWGN_results}), \OurAlg demonstrates consistent advantages across the entire SNR range: at low SNR (0--10~dB), it achieves {+0.80 to +1.28~dB PSNR improvement} over DeepJSCC, while at high SNR (25--30~dB), the gap narrows to {+0.30 to +0.60~dB} (e.g., 29.80~dB vs. 29.20~dB at 30~dB), confirming that \OurAlg's robustness mechanisms do not significantly compromise performance under favorable channel conditions. This pattern stems from the channel-level WDRO formulation, which optimizes for worst-case perturbations, providing substantial benefits in challenging propagation regimes while converging to semantic encoder-decoder capacity limits at high SNR where channel effects become negligible.} 

\blue{
The trends are observed Under Rayleigh fading (Fig.~\ref{fig:combined_2}, Table~\ref{tab:Rayleigh_results}), where \OurAlg demonstrates larger
advantages at low-to-moderate SNR, validating robustness to time-varying multipath channels. When semantic perturbations are introduced (FGSM with 10\% noise), DeepJSCC experiences degradation in both PSNR and SSIM, particularly under Rayleigh fading. DeepSC-RI mitigates this drop effectively due to its masked representation design, but still suffers at lower SNRs. In contrast, \OurAlg maintains more stable performance across the SNR range, with the performance gap widening further under adversarial conditions. This behavior validates our bilevel WDRO design: while DeepSC-RI enforces robustness at the semantic encoding level, \OurAlg jointly addresses uncertainty in both semantic and channel domains through independent Wasserstein balls. This dual consideration enables adaptive encoding that maintains semantic fidelity even under the combined stress of input perturbations and channel degradation.}

\begin{figure}[t]
    \centering
    \begin{subfigure}{\linewidth}
        \centering
        \includegraphics[width=0.85\linewidth]{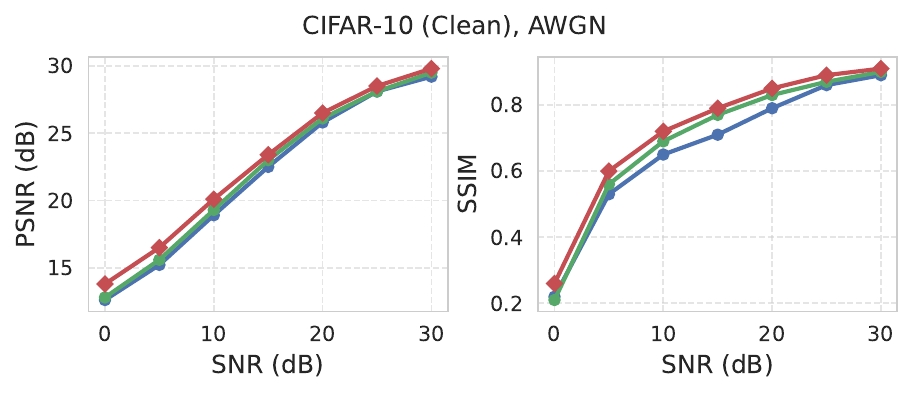}
        \label{fig:acc}
    \end{subfigure}
    
    \vspace{-10pt} 
    
    \begin{subfigure}{\linewidth}
        \centering
        \includegraphics[width=0.85\linewidth]{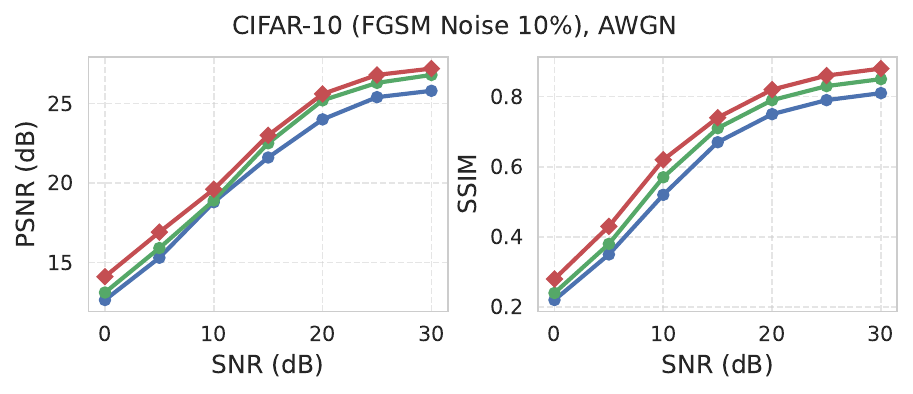}
        \label{fig:psnr10}
    
    \end{subfigure}
    \vspace{-10pt} 
    
    \begin{subfigure}{\linewidth}
        \centering
        \includegraphics[width=0.85\linewidth]{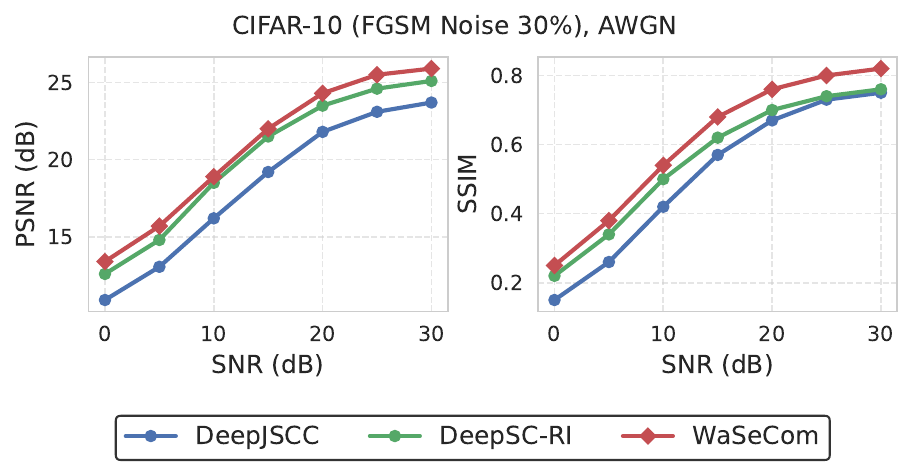}
        \label{fig:psnr30}
    \end{subfigure}
    \caption{Performance of image transmission tasks with different semantic noise ratio under AWGN channel}
    \label{fig:combined}
\end{figure}

\begin{figure}[h]
    \centering
    \begin{subfigure}{\linewidth}
        \centering
        \includegraphics[width=0.85\linewidth]{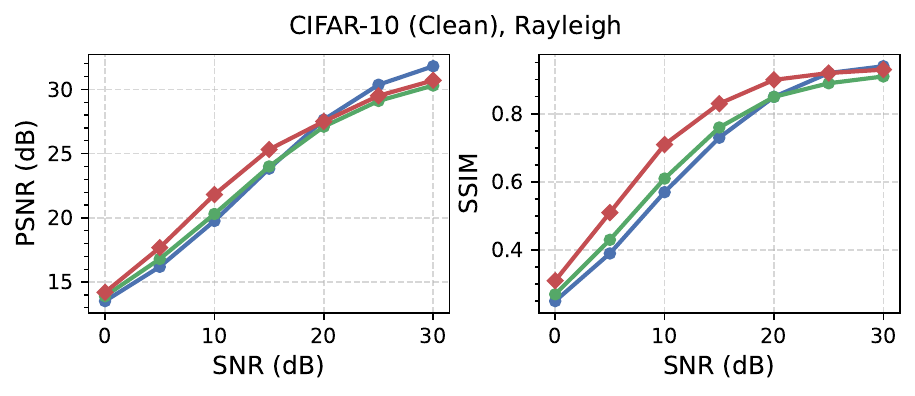}
        \label{fig:ray_leigh_acc}
    \end{subfigure}
    
    \vspace{-10pt} 

    \begin{subfigure}{\linewidth}
        \centering
        \includegraphics[width=0.85\linewidth]{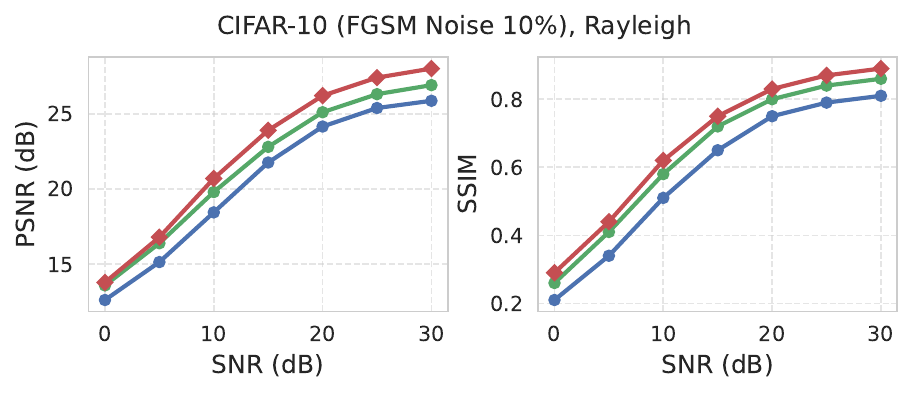}
        \label{fig:ray_leigh_psnr10}
    \end{subfigure}
    
    \vspace{-10pt} 

    \begin{subfigure}{\linewidth}
        \centering
        \includegraphics[width=0.85\linewidth]{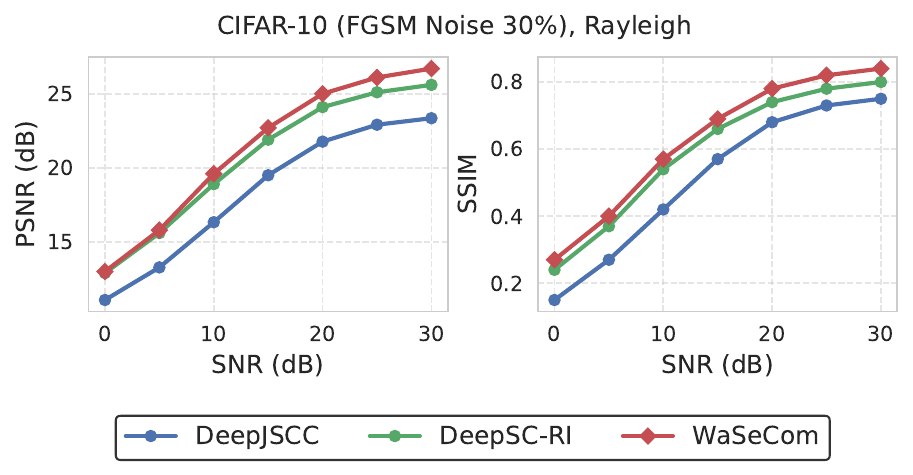}
        \label{fig:ray_leigh_psnr30}
    \end{subfigure}

    \caption{Performance of image transmission tasks with different semantic noise ratio under Rayleigh channel}
    \label{fig:combined_2}
\end{figure}

Under stronger semantic noise (FGSM with 30\% noise ratio), the advantage of \OurAlg becomes more visible. DeepJSCC’s performance degrades rapidly, particularly below 10~dB SNR. This confirms its vulnerability to out-of-distribution inputs. DeepSC-RI remains more stable but begins to plateau, while \OurAlg demonstrates lower degradation rates across all SNRs. This resilience stems from its ability to optimize under worst-case semantic and channel shifts, resulting in a flatter degradation curve. For example, at 10~dB SNR under Rayleigh fading, \OurAlg retains considerable higher PSNR and SSIM compared to both baselines.

Channel variability has a marked impact on all methods, especially under Rayleigh fading. The performance gap between AWGN and Rayleigh conditions is widest at low SNRs, where multipath fading dominates. Although robust to semantic shifts, DeepSC-RI often lacks channel-aware training and thus underperforms in fading scenarios. In contrast, \OurAlg  outperforms both baselines in Rayleigh channels by an observable amount, highlighting the effectiveness of integrating channel uncertainty into the optimization objective. The WDRO formulation enables \OurAlg to anticipate adversarial channel conditions, leading to smoother performance curves across diverse environments.

\blue{A quantitative comparison (Tables~\ref{tab:AWGN_results} and~\ref{tab:Rayleigh_results}) reveals that WaSeCom provides superior robustness under challenging conditions. At low SNRs, its resilience is evident: under the AWGN channel with 30\% FGSM noise at 0 dB, WaSeCom achieves 13.40 dB, a notable improvement over DeepSC-RI (12.60 dB) and DeepJSCC (10.91 dB). This advantage is most pronounced at high SNRs under strong adversarial attack. For example, in the Rayleigh channel at 30 dB with 30\% FGSM noise, WaSeCom maintains 26.70 dB, whereas DeepSC-RI and DeepJSCC degrade to 25.60 dB and 23.35 dB, respectively.} The empirical trends align with the theoretical motivations of \OurAlg. By modeling the robustness problem as a bi-level optimization over Wasserstein balls, the model is encouraged to generalize beyond nominal data distributions. The inner semantic-level objective regularizes feature encoding against worst-case input shifts, while the outer channel-level optimization ensures that these features are decodable under stochastic degradation. This layered robustness yields improvements in different noise settings and contributes to stable performance under mild perturbations and channel variance.

\subsubsection{Performance on Text Semantic Communication}

We evaluate the performance of our proposed \OurAlg framework for text semantic transmission under a variety of wireless channel conditions using the BLEU score. \OurAlg is compared with DeepSC~\cite{deepsc}, a state-of-the-art Transformer-based SemCom model that focuses on maximizing semantic fidelity but lacks robustness mechanisms against channel and input perturbations.

Fig.~\ref{fig:BLEU_SNR} presents the BLEU scores of \OurAlg and DeepSC under Rayleigh fading and AWGN channels across a range of signal-to-noise ratios (0dB to 18dB). Under Rayleigh fading, \OurAlg demonstrates superior robustness at low SNRs. For example, at 0dB, \OurAlg achieves a BLEU score of approximately 0.55 compared to DeepSC’s 0.50, with the performance gap remaining consistent (0.03–0.05 points) across all SNR levels. A similar trend is observed under AWGN, where \OurAlg again maintains a 0.05 BLEU point advantage, achieving 0.50 at 0dB while DeepSC trails at 0.45. Although the absolute improvements are modest, \OurAlg consistently delivers more stable and resilient performance in both channel types, confirming the benefit of its bi-level training strategy that explicitly optimizes for worst-case semantic and channel perturbations.

\begin{table*}[t]
\centering
\scalebox{0.83}{
\begin{tabular}{|c||ccc|ccc|ccc|}
\hline
\multirow{2}{*}{\textbf{SNR (dB)}}
& \multicolumn{3}{c|}{\textbf{Noise-Free}} 
& \multicolumn{3}{c|}{\textbf{FGSM Noise 10\%}}
& \multicolumn{3}{c|}{\textbf{FGSM Noise 30\%}} \\
\cline{2-10}
& DeepJSCC & DeepSC-RI & \OurAlg 
& DeepJSCC & DeepSC-RI & \OurAlg
& DeepJSCC & DeepSC-RI & \OurAlg \\ 
\hline\hline
0  & 12.61 & 12.80 & \underline{13.80} 
   & 12.63 & 13.10 & \underline{14.10}
   & 10.91 & 12.60 & \underline{13.40} \\
5  & 15.22 & 15.60 & \underline{16.50} 
   & 15.29 & 15.90 & \underline{16.90}
   & 13.06 & 14.80 & \underline{15.70} \\
10 & 18.91 & 19.30 & \underline{20.10} 
   & 18.79 & 18.90 & \underline{19.60}
   & 16.20 & 18.50 & \underline{18.90} \\
15 & 22.50 & 23.00 & \underline{23.40} 
   & 21.60 & 22.50 & \underline{23.00}
   & 19.20 & 21.50 & \underline{22.00} \\
20 & 25.80 & 26.10 & \underline{26.50} 
   & 24.00 & 25.20 & \underline{25.60}
   & 21.80 & 23.50 & \underline{24.30} \\
25 & 28.10 & 28.10 & \underline{28.50} 
   & 25.40 & 26.30 & \underline{26.80}
   & 23.10 & 24.60 & \underline{25.50} \\
30 & 29.20 & 29.50 & \underline{29.80} 
   & 25.80 & 26.80 & \underline{27.20}
   & 23.70 & 25.10 & \underline{25.90} \\
\hline
\end{tabular}
}
\caption{\blue{PSNR (dB) comparison under Clean and FGSM attacks over the {AWGN} channel on CIFAR-10.}}
\label{tab:AWGN_results}
\end{table*}

\begin{table*}[t]
\centering
\scalebox{0.83}{
\begin{tabular}{|c||ccc|ccc|ccc|}
\hline
\multirow{2}{*}{\textbf{SNR (dB)}}
& \multicolumn{3}{c|}{\textbf{Noise-Free}} 
& \multicolumn{3}{c|}{\textbf{FGSM Noise 10\%}} 
& \multicolumn{3}{c|}{\textbf{FGSM Noise 30\%}} \\
\cline{2-10}
& DeepJSCC & DeepSC-RI & \OurAlg 
& DeepJSCC & DeepSC-RI & \OurAlg
& DeepJSCC & DeepSC-RI & \OurAlg \\ 
\hline\hline
0  & 13.52 & 13.90 & \underline{14.20}
   & 12.63 & 13.60 & \underline{13.80}
   & 11.08 & 12.90 & \underline{13.00} \\
5  & 16.19 & 16.80 & \underline{17.68}
   & 15.15 & 16.40 & \underline{16.80}
   & 13.28 & 15.60 & \underline{15.80} \\
10 & 19.76 & 20.30 & \underline{21.81}
   & 18.45 & 19.80 & \underline{20.70}
   & 16.33 & 18.90 & \underline{19.60} \\
15 & 23.83 & 24.00 & \underline{25.32}
   & 21.76 & 22.80 & \underline{23.90}
   & 19.50 & 21.90 & \underline{22.70} \\
20 & 27.64 & 27.10 & \underline{27.65}
   & 24.15 & 25.10 & \underline{26.20}
   & 21.77 & 24.10 & \underline{25.00} \\
25 & \underline{30.36} & 29.10 & {29.50}
   & 25.38 & 26.30 & \underline{27.40}
   & 22.91 & 25.10 & \underline{26.10} \\
30 & \underline{31.80} & 30.30 & {30.70}
   & 25.86 & 26.90 & \underline{28.00}
   & 23.35 & 25.60 & \underline{26.70} \\
\hline
\end{tabular}
}
\caption{\blue{PSNR (dB) comparison under Clean and FGSM attacks over the {Rayleigh} channel on CIFAR-10.}}
\label{tab:Rayleigh_results}
\end{table*}

To further evaluate robustness, we test both models under FGSM adversarial attacks with a perturbation strength of 10\%. Fig.~\ref{fig:BLEU_SNR_FGSM} shows BLEU scores under adversarially perturbed inputs across both fading scenarios. \blue{Under Rayleigh fading, WaSeCom exhibits a sharp performance advantage: it reaches approximately $0.75$ BLEU at $5$dB and stabilizes in the $0.90–0.92$ range at an SNR of approximately $20$dB. This widens its performance gap over DeepSC, which peaks lower in the $0.83–0.85$ range.} The gap widens to 0.10–0.15 BLEU points at 30dB, highlighting \OurAlg’s greater resilience to adversarial semantic distortions in harsh channel conditions. \blue{Under the AWGN channel, WaSeCom shows a distinct performance advantage in the low-to-mid SNR regime ($0-12$ dB). At SNRs of $13$ dB and above, both models reach a performance ceiling, achieving comparable near-perfect BLEU scores. This demonstrates that WaSeCom's robustness is most impactful in challenging channel conditions without sacrificing performance in high-quality channels.} \blue{Under the Rayleigh channel, WaSeCom's performance is highly competitive, showing a slight advantage at low SNRs (0-3 dB) while performing comparably to DeepSC in the mid-SNR range (6-12 dB). Under the AWGN channel, however, WaSeCom demonstrates a clear and consistent performance advantage across all tested SNRs. This nuanced result confirms its effective generalization, particularly showcasing its robustness in less variable channel conditions.}
\blue{The performance dip for \OurAlg observed in the 9–12 dB SNR region of Fig.~\ref{fig:BLEU_SNR_FGSM}a is an expected consequence of the robustness–fidelity trade-off inherent in WDRO. Within this intermediate SNR range, our framework prioritizes worst-case stability, allocating model capacity to defend against distributional shifts. In contrast, DeepSC, which is optimized solely for average-case performance, can achieve slightly higher fidelity in this narrow band. However, The performance of \OurAlg is better at both lower ($<$9 dB) and higher ($>$12 dB) SNRs, which demonstrates its greater overall resilience, which is the primary objective of its design.} 

Overall, these results validate \OurAlg's design objective of achieving dual robustness. It helps maintain semantic fidelity in the face of both channel degradation and input-level adversarial shifts. While improvements over DeepSC are not always large in absolute terms, they are consistent, particularly in the more challenging Rayleigh fading and adversarial settings. This makes \OurAlg a more reliable solution for real-world text SemCom systems.

\begin{figure}
    \centering	    
    \includegraphics[width=\columnwidth]{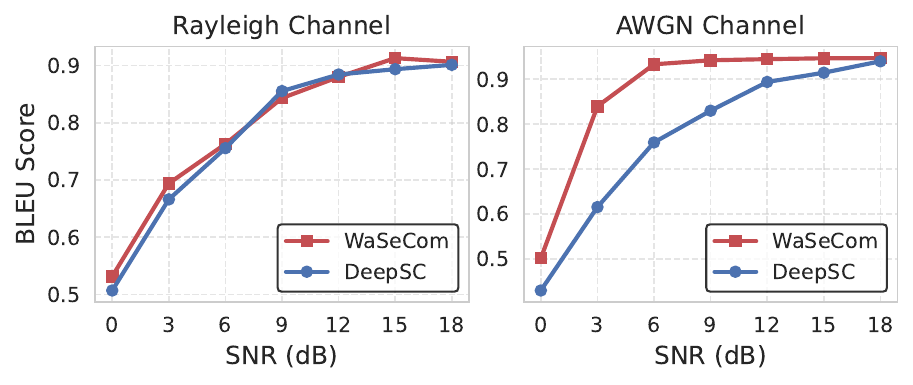}
    \caption{Performance of text transmission without semantic noise under different channels}
    \label{fig:BLEU_SNR}
\end{figure}

\begin{figure}
    \centering
    \includegraphics[width=\columnwidth]{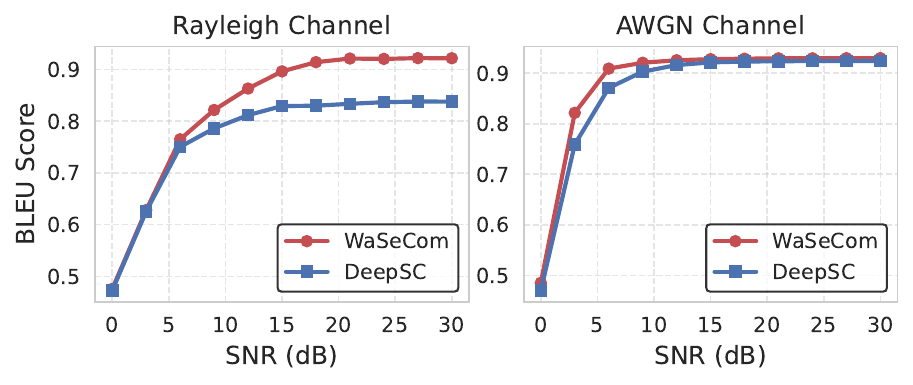}
    \caption{Performance of text transmission with semantic noise (added by an FGSM adversarial attack) under different channels}
    \label{fig:BLEU_SNR_FGSM}
\end{figure}

\blue{
\subsection{Ablation Studies}
\label{sec:ablation}

\subsubsection{Convergence Results}
\label{sec:conv_results}

To empirically validate the convergence properties established in Section~\ref{sec:generalization}, we analyze the training dynamics of Algorithm~\ref{alg:wasecom} over 100 epochs on the CIFAR-10 dataset under FGSM adversarial perturbations (10\% noise).
Fig.~\ref{fig:conv} illustrates the convergence behavior for both PSNR (Fig.~\ref{fig:conv_psnr}) and SSIM (Fig.~\ref{fig:conv_ssim}) metrics. Both exhibit three distinct phases: \textit{(i) Rapid initial improvement} (epochs 0--20), where PSNR increases from approximately 14 dB to 25 dB and SSIM rises from 0.20 to 0.80, reflecting efficient optimization through the initial loss landscape; \textit{(ii) Gradual refinement} (epochs 20--50), characterized by continued but decelerating improvement as the algorithm approaches a stationary point; and \textit{(iii) Stable convergence} (epochs 50--100), where metrics plateau with minimal oscillation—PSNR stabilizes around 27--28 dB and SSIM near 0.90--0.91.

Critically, the training and validation curves remain tightly aligned throughout all phases, with maximum deviation $< 0.3$ dB for PSNR and $< 0.02$ for SSIM. This tight coupling indicates \textit{no overfitting}, a direct consequence of the \OurAlg's objective that inherently regularizes against distributional shifts by optimizing over Wasserstein balls rather than the empirical distribution alone. The smooth, monotonic convergence without significant oscillations validates the stability of the bilevel alternating optimization scheme, and thus confirming that the inner and outer loops sufficiently coordinate to find a robust stationary solution.


\begin{figure}[t]
    \centering
    \begin{subfigure}{\linewidth}
		\centering	    
        \includegraphics[width=0.65\linewidth]{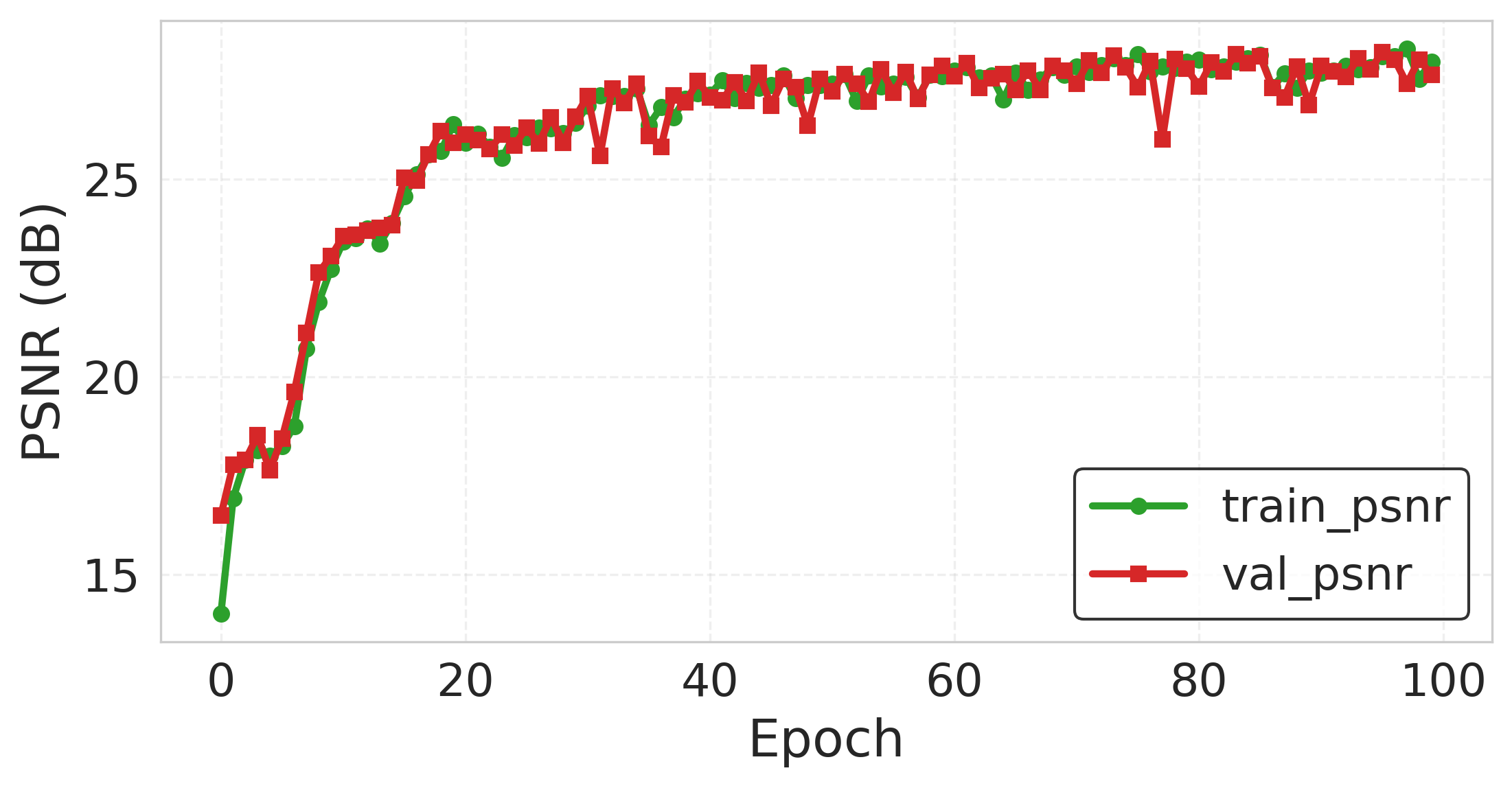}
        \caption{PSNR convergence}
        \label{fig:conv_psnr}
    \end{subfigure}
    \begin{subfigure}{\linewidth}
        \centering
        \includegraphics[width=0.65\linewidth]{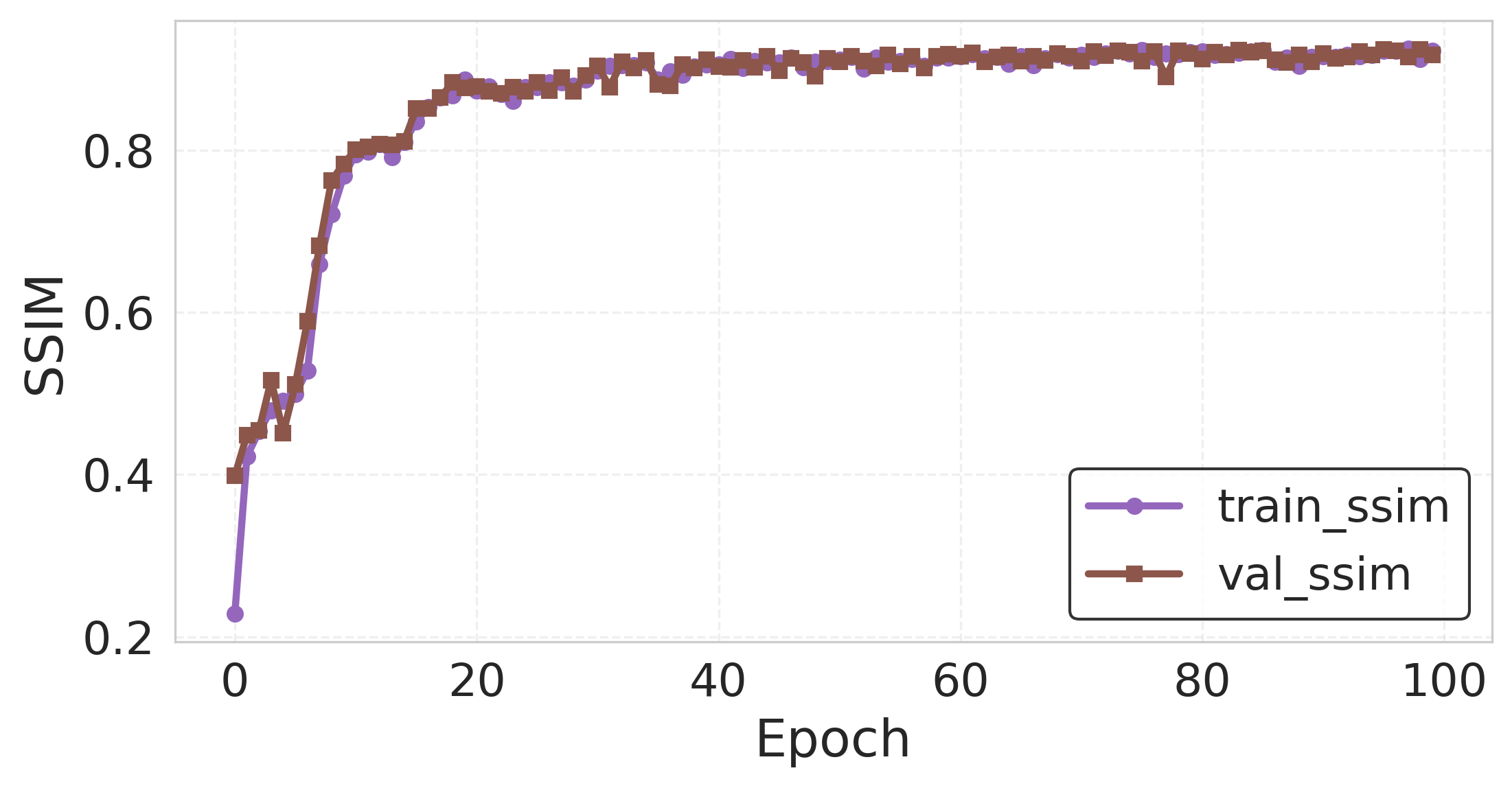}
        \caption{SSIM convergence}
        \label{fig:conv_ssim}
    \end{subfigure}
    \caption{\blue{Empirical convergence analysis of \OurAlg in terms of (a) PSNR and (b) SSIM.}}
    \label{fig:conv}
\end{figure}

\subsubsection{Sensitivity Analysis}
\blue{
To evaluate the robustness of \OurAlg to hyperparameter selection, we conduct a comprehensive sensitivity analysis of the Wasserstein radii $\rho$ and $\mu$ under adversarial semantic noise (FGSM with 10\% noise). This analysis quantifies how each parameter affects the robustness-fidelity trade-off and identifies optimal operating regions.}

\blue{
Table~\ref{tab:mu_sweep} reveals a clear non-monotonic relationship between the channel-level Wasserstein radius $\mu$ and reconstruction quality. When $\mu$ is undersized (e.g., $\mu$ = 0.005), the model exhibits insufficient robustness to channel-level distributional shifts, resulting in degraded performance (PSNR $\approx$ 21.97 dB and SSIM $\approx$ 0.75). As $\mu$ increases to 0.01, both metrics reach their peak (PSNR $\approx$ 23.33 dB and SSIM $\approx$ 0.80), representing the best balance between robustness to channel uncertainty and reconstruction fidelity under nominal conditions. However, further increases lead to monotonic degradation: at $\mu$ = 0.5, PSNR drops by 0.81 dB compared to the optimum. This decline occurs because excessively large ambiguity sets force the model to optimize for unrealistically severe channel perturbations, resulting in overly conservative representations that sacrifice nominal reconstruction quality.}

\blue{Table~\ref{tab:rho_sweep} demonstrates that the semantic-level radius $\rho$ exhibits a similar non-monotonic trend, but with a broader optimal region. At $\rho$ = 0.005, insufficient regularization yields suboptimal robustness (PSNR $\approx$ 23.23 dB and SSIM $\approx$ 0.79). Performance improves as $\rho$ increases, reaching its optimum at $\rho$ = 0.05 (PSNR $\approx$ 23.93 dB and SSIM $\approx$ 0.80), representing a gain of 0.604 dB over $\rho$ = 0.01. This improvement occurs because moderate expansion of the semantic ambiguity set enables learning of more robust feature representations that explicitly account for adversarial input shifts. Beyond this optimum, performance degrades gradually: even at $\rho$ = 0.5, the model retains reasonable performance (PSNR $\approx$ 23.53 dB). This gentler degradation compared to $\mu$ suggests semantic-level robustness is more forgiving to overestimation, though excessively large $\rho$ values still induce over-smoothing that reduces fine-grained semantic information.
These results demonstrate that \OurAlg achieves stable performance across a reasonable range of hyperparameters ($\mu \in [0.01, 0.05]$ and $\rho \in [0.01, 0.05]$), with optimal settings at $\mu = 0.01$ and $\rho = 0.05$ under this adversarial scenario. The relatively narrow sensitivity regions confirm the importance of proper calibration while indicating that the framework does not require extremely fine-tuned hyperparameters to achieve robust performance.}

\begin{table}[t]
\centering
\begin{tabular}{|c|c|cc|}
\hline
$\rho$ & $\mu$ & PSNR & SSIM \\
\hline
\hline
\multirow{6}{*}{0.01} & 0.005 & 21.986 & 0.747 \\
 & \underline{0.010} & \underline{23.328} & \underline{0.798} \\
 & 0.050 & 22.771 & 0.758 \\
 & 0.100 & 22.675 & 0.754 \\
 & 0.200 & 22.598 & 0.750 \\
 & 0.500 & 22.522 & 0.746 \\
\hline
\end{tabular}
\caption{\blue{Sensitivity to channel-level radius $\mu$  with $\rho{=}0.01$ under FGSM with 10\% noise. Best trade-off near $\mu{=}0.01$.}}
\label{tab:mu_sweep}
\end{table}

\begin{table}[t]
\centering
\begin{tabular}{|c|c|cc|}
\hline
$\mu$ & $\rho$ & PSNR & SSIM \\
\hline
\hline
\multirow{6}{*}{0.01} & 0.005 & 23.234 & 0.785 \\
 & 0.010 & 23.328 & 0.798 \\
 & \underline{0.050} & \underline{23.932} & \underline{0.802} \\
 & 0.100 & 23.684 & 0.793 \\
 & 0.200 & 23.596 & 0.789 \\
 & 0.500 & 23.532 & 0.786 \\
\hline
\end{tabular}
\caption{\blue{Sensitivity to semantic-level radius $\rho$ with $\mu{=}0.01$ under FGSM with 10\% noise.  Moderate $\rho$ (e.g., $0.05$) yields the best robustness–fidelity balance.}}
\label{tab:rho_sweep}
\end{table}

\blue{
To provide a more comprehensive assessment of hyperparameter sensitivity, Fig.~\ref{fig:fgsm_mu} and Fig.~\ref{fig:fgsm_rho} present the performance of \OurAlg across a wide range of SNR conditions, from severe noise ($-10$ dB) to high-quality channels ($30$ dB), while systematically varying the Wasserstein radii under FGSM 10\% adversarial perturbations.}

\blue{
The results in Fig\ref{fig:fgsm_mu} demonstrated remarkable stability of \OurAlg with respect to $\mu$ across the entire SNR spectrum. All tested values of $\mu \in [0.005, 0.5]$ produce nearly overlapping PSNR and SSIM curves, with maximum deviations limited to approximately $0.5$ dB in PSNR across all SNR regimes. This consistency indicates that once $\mu$ is set within a reasonable range, the framework maintains robust performance regardless of channel quality. The convergence of all curves at high SNR ($>20$ dB) suggests that channel-level robustness becomes less critical when transmission conditions are favorable, while the maintained separation at low SNR confirms that proper $\mu$ calibration provides meaningful protection under severe noise. Notably, even the extreme setting of $\mu = 0.5$ does not catastrophically degrade performance, deviating by less than $1$ dB from the optimal configuration across most conditions.}

\blue{Similarly, Fig.~\ref{fig:fgsm_rho} reveals that variations in $\rho$ over two orders of magnitude ($0.005$ to $0.5$) produce tightly clustered performance curves for both PSNR and SSIM across all SNR levels. The maximum spread among different $\rho$ values remains within approximately $0.8$ dB throughout the SNR range, with all configurations following nearly identical trends. This robustness to $\rho$ selection is particularly evident in the mid-to-high SNR regime ($5$–$30$ dB), where the curves are virtually indistinguishable, suggesting that semantic-level robustness mechanisms are effective across a broad hyperparameter range. At very low SNR ($<0$ dB), a slight divergence emerges, with moderate $\rho$ values ($0.05$–$0.1$) showing marginally better performance, consistent with the findings in Table~\ref{tab:rho_sweep}. The graceful degradation pattern across all $\rho$ settings confirms that \OurAlg does not exhibit sharp performance cliffs, making it practical  without exhaustive hyperparameter tuning.}

\blue{ These SNR-sweep experiments provide strong empirical evidence that \OurAlg is not critically sensitive to exact hyperparameter selection within reasonable bounds. The consistent performance across $\mu \in [0.01, 0.1]$ and $\rho \in [0.01, 0.1]$ for all channel conditions validates the framework's practical applicability and suggests that practitioners can achieve near-optimal performance with moderate hyperparameter search efforts. This stability is particularly valuable in real-world wireless systems where channel conditions vary dynamically and retuning hyperparameters for each scenario is infeasible.}

\begin{figure}[t]
\centering
\includegraphics[width=1.0\linewidth]{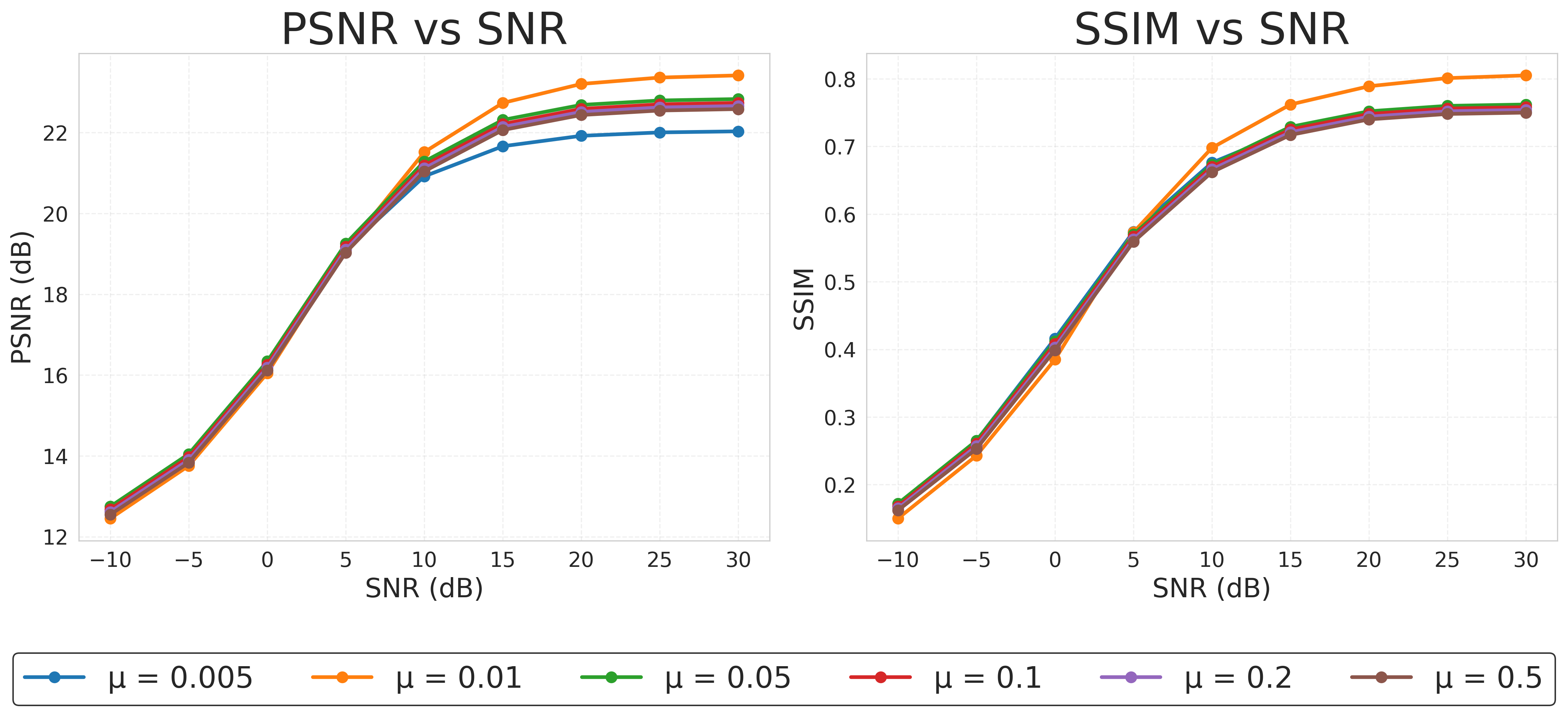}
\caption{\blue{Sensitivity of \OurAlg to the {channel-level radius} $\mu$ under FGSM Noise 10\% with fixed $\rho{=}0.01$. 
PSNR and SSIM show similar trends across $\mu\in[0.005,0.5]$, with only marginal differences ($<$ 0.5 dB PSNR), indicating robustness to $\mu$ selection.}}
\label{fig:fgsm_mu}
\end{figure}
\begin{figure}[t]
\centering
\includegraphics[width=1.0\linewidth]{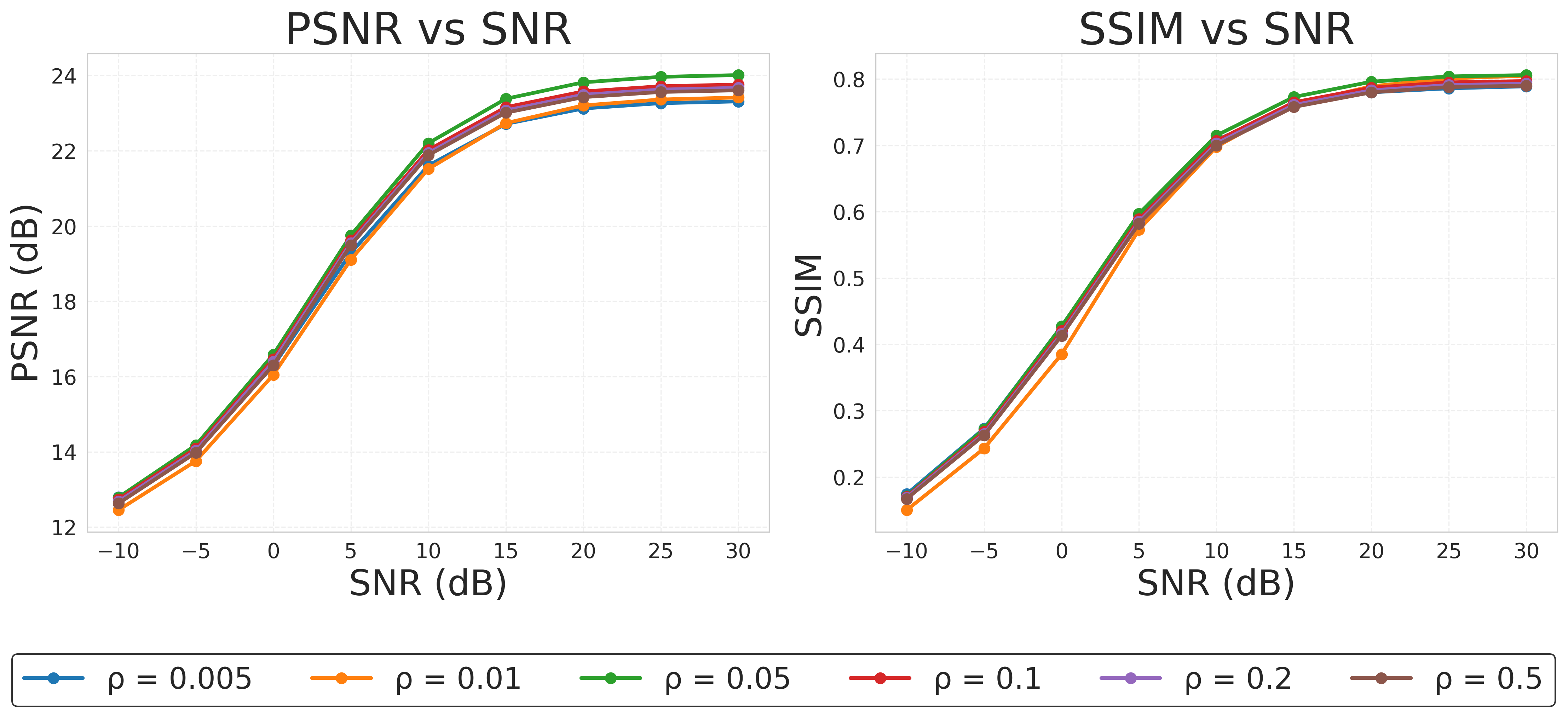}
\caption{\blue{Sensitivity of \OurAlg to the {semantic-level radius} $\rho$ under FGSM Noise 10\% with fixed $\mu{=}0.01$. 
Performance remains stable, suggesting that moderate changes in $\rho$ have little influence on the robustness–fidelity trade-off.}}
\label{fig:fgsm_rho}
\end{figure}
}
\section{Conclusion}
\label{sec:conclusion}


In this paper, we introduce \OurAlg, a novel framework designed to enhance the robustness of wireless SemCom in the presence of semantic and channel-level uncertainties. By formulating the problem within a bilevel, distributionally robust optimization problem, our approach mitigates semantic interpretation errors and transmission distortions. Theoretical analyses establish robustness guarantees under distributional shifts at both semantic and physical layers, while extensive experiments across multiple modalities and communication settings demonstrate consistent performance improvements over state-of-the-art baselines. These findings underscore the efficacy of integrating principled, worst-case optimization into next-generation task-oriented wireless communication. \blue{Despite its strengths, the framework presents limitations that guide future research. First, the bilevel optimization incurs extra computational overhead during offline training. Future work will therefore focus on exploring more efficient surrogate formulations to reduce this training complexity. Second, a comprehensive evaluation under dynamic, time-varying channel conditions is needed to fully validate the framework's practical applicability in mobile scenarios.}

\bibliographystyle{IEEEtran}
\bibliography{references}

\newpage
\appendix
\subsection{Proof of Lemma~\ref{lem:excess_risk}}
\label{proof:excess_risk}

We present a detailed proof of Lemma~\ref{lem:excess_risk}, which establishes that the robust surrogate excess risks derived from the dual WDRO formulation closely approximate the worst-case excess risks under distributional shifts. The result holds for both the semantic-level and channel-level objectives in the proposed bi-level framework in \OurAlg.

In what follows, we unify both levels and present the proof for a generic WDRO problem. Let denote a generic input as \( u \in \mathcal{U} \), where \( \mathcal{U} = \mathcal{X} \) in the semantic level and \( \mathcal{U} = \mathcal{Z} \) in the channel level. Let \( P \) be a probability distribution over \( \mathcal{U} \), and let \( \ell(u; f) \) denote the loss function under hypothesis \( f \in \mathcal{F} \), which is assumed to be \(L\)-Lipschitz in \( u \). Specifically:
\begin{itemize}
  \item For the semantic level, \( u = x \in \mathcal{X} \), \( \ell(u; f) = \ell_{\text{s}}(x; \vartheta) \), where \( f \) parameterizes the semantic encoder-decoder pair \( \vartheta \).
  \item For the channel level, \( u = z \in \mathcal{Z} \), \( \ell(u; f) = \ell_{\text{c}}(s, z; \varphi) \), where \( f \) represents the channel encoder-decoder pair \( \varphi \), and \( s = f_\theta(x) \) is fixed.
\end{itemize}



Let \( Q \in \mathcal{B}_p(P, \rho) \) be a distribution within a \(p\)-Wasserstein ball of radius \( \rho \). We define the robust surrogate risk based on the dual formulation:
\begin{align*}
\phi_\gamma(u; f) &:= \sup_{\tilde{u} \in P} \left\{ f(\tilde{u}) - \gamma \|\tilde{u} - u\|^2 \right\}, \\
\mathscr{L}_\rho^\gamma(P, f) &:= \mathbb{E}_{u \sim P}[\phi_\gamma(u; f)] + \gamma \rho^2.
\end{align*}

We aim to bound the difference between \( \mathscr{E}(Q, f) \) and \( \mathscr{E}_\rho^\gamma(P, f) \), where:
\begin{align}
\mathscr{E}(Q, f) &:= \mathscr{L}(Q, f) - \inf_{f' \in \mathcal{L}} \mathscr{L}(Q, f'), \\
\mathscr{E}_\rho^\gamma(P, f) &:= \mathscr{L}_\rho^\gamma(P, f) - \inf_{f' \in \mathcal{L}} \mathscr{L}_\rho^\gamma(P, f'). 
\end{align}

We first prove the following fact:
	
\textbf{Fact 1}: Surrogate risk upper-bounds the true risk.
\begin{align*}
	\vspace{-1mm}
	\quad (a) \quad &\mathscr{L}(Q, f) \leq  \mathscr{L}_{\rho}^{\gamma}(P, f), \qquad \forall f \in \mathcal{L},  Q \in \gB(P, \rho).   \\
	 (b) \quad  &\inf_{f' \in \mathcal{L}} \mathscr{L}(Q, f') \leq \inf_{f' \in \mathcal{L}} \mathscr{L}_{\rho}^{\gamma}(P, f') , \quad \forall Q \in \gB(P, \rho).  
	\end{align*}
	
	For (a), we have
	\begin{align*}
	&\mathscr{L}(Q,f) \\
	& \leq  \underset{P' \in \gB({P}, \rho)}{\mbox{sup }} \mathscr{L} (P',f) =  \, \underset{\gamma' \geq 0}{{ \inf }} \Big{\{} \gamma' \rho^2 + \mathbf{E}_{u \sim {P}}\Big{[}\phi_\gamma(u,  f)\Big{]}  \Big{\}} \\
	&  \leq   \gamma \rho^2 + \mathbf{E}_{x \sim {P}}\Big{[}\phi_\gamma(u,  f)\Big{]}  \eqdef \mathscr{L}_{\rho}^{\gamma}(P,f),
	\end{align*}
	where the equality is due to strong  duality result by Gao et al.~\cite{gao_distributionally_2016}. 

	For (b), defining  $f_{P} \defeq \argmin_{f' \in \mathcal{L}} \mathscr{L}_{\rho}^{\gamma}(P, f')$, we have 
	\begin{align} 
	\inf_{f' \in \mathcal{L}} \mathscr{L}(Q, f') & \leq \mathscr{L}(Q,f_{P}) \\
    &  \leq  \underset{P' \in \gB({P}, \rho)}{\mbox{sup }} \mathscr{L} (P',f_{P})  \\ 
	& = \underset{\gamma' \geq 0}{{ \inf }} \Big{\{} \gamma' \rho^2 + \mathbf{E}_{x \sim {P}}\Big{[}\phi_\gamma(u,  f_{P})\Big{]}  \Big{\}} \\ 
	& \leq \gamma \rho^2 + \mathbf{E}_{x \sim {P}}\Big{[}\phi_\gamma(u,  f_{P})\Big{]} \\
	& = \inf_{f' \in \mathcal{L}} \mathscr{L}_{\rho}^{\gamma}(P, f'). 
	\end{align}
	
	We next prove the second fact:
	
	\textbf{Fact 2:} Surrogate risk is close to true risk.
	\begin{align*}
	\vspace{-1mm}
	\quad (a) \quad &\mathscr{L}_{\rho}^{\gamma}(P,f) \leq  \mathscr{L}(Q,f) + 2L\rho   + \abs{\gamma - \gamma^*} {\rho}^2, \\
	& \myquad[10] \forall f \in \gF, Q \in \gB(P, \rho)   \\
	(b) \quad  &\inf_{f' \in \mathcal{L}} \mathscr{L}_{\rho}^{\gamma}(P, f') \leq   \inf_{f' \in \mathcal{L}} \mathscr{L}(Q, f') + 2L \rho   + \abs{\gamma - \gamma^*} {\rho}^2.
	\end{align*}
	
	For (a), we have:
	\begin{align*} 
	&\mathscr{L}_{\rho}^{\gamma}(P,f) \\  
	& =  \biggC{ \underset{P' \in \gB({P}, \rho)}{\mbox{sup }} \mathscr{L} (P',f)} + \BigC{\mathscr{L}_{\rho}^{\gamma}(P,f)  - \underset{P' \in \gB({P}, \rho)}{\mbox{sup }} \mathscr{L} (P',f)} \\
	&\leq  \BigC{ \mathscr{L}(Q,f) + 2L\rho} + \biggC{\mathbf{E}_{u \sim P}[\phi_\gamma(u, f)]  + {\rho}^2\gamma \\
	& \myquad[12] - \min_{\gamma^{\prime} \geq 0} \Big\{ {\rho}^2 \gamma^{\prime} + \mathbf{E}_{u \sim P}[\phi_{\gamma'}(u, f)] \Big\}} \\
	& \leq   \mathscr{L}(Q,f) + 2L\rho +  {\rho}^2(\gamma - \gamma^*) + \mathbf{E}_{u \sim P}\Big[\phi_{\gamma}(u, f) - \phi_{\gamma^*}(u, f)\Big] \\
	& =  \begin{aligned}[t] & \mathscr{L}(Q,f) + 2L\rho + {\rho}^2(\gamma - \gamma^*) \\
		&+ \mathbf{E}_{u \sim P}\bigg[ \sup_{\zeta \in \mathcal{Z}} \Big\{ \ell(\zeta,h)  - \gamma d(\zeta,  u) \Big\}  - \sup_{\zeta \in \mathcal{Z}} \Big\{ \ell(\zeta,h) - \gamma^* d^2(\zeta,  u)  \Big\} \bigg] \end{aligned} \\
	& =  \mathscr{L}(Q,f) + 2L\rho +  (\gamma - \gamma^*) \BigP{{\rho}^2 - \mathbf{E}_{u \sim P}\Big[\sup_{\zeta \in \mathcal{Z}} d^2(\zeta,  u)\Big] }\\
	& \leq  \mathscr{L}(Q,f) + 2L\rho  + \abs{\gamma - \gamma^*}  {\rho}^2,
	\end{align*}
	where the first inequality is due to Proposition~\ref{pro1}, and the last inequality is because we choose $\gamma \geq L/\rho$ and that fact that $\gamma^* \leq L/\rho$ by Lemma 1 of~\cite{lee_minimax_2018}. 


%
	
	For (b), defining  $f_Q \defeq \argmin_{f \in \mathcal{L}} \mathscr{L}(Q, f)$, we have
	
	\begin{align}
	\inf_{f' \in \mathcal{L}} \mathscr{L}_{\rho}^{\gamma}(P, f') &\leq  \mathscr{L}_{\rho}^{\gamma}(P, f_Q)  \\
	&\leq \mathscr{L}(Q, f_{Q}) + 2L \rho   + \abs{\gamma - \gamma^*} {\rho}^2 \\
	&= \inf_{f' \in \mathcal{L}} \mathscr{L}(Q, f') + 2L \rho   + \abs{\gamma - \gamma^*} {\rho}^2,
	\end{align}
	
	where the second line is due to \textbf{Fact 2}(a). 

Combining the bounds:

Let us denote the robust excess risk and the worst-case excess risk as:
\begin{align*}
\mathscr{E}_\rho^\gamma(P, f) &= \mathscr{L}_\rho^\gamma(P, f) - \inf_{f'} \mathscr{L}_\rho^\gamma(P, f'), \\
\mathscr{E}(Q, f) &= \mathscr{L}(Q, f) - \inf_{f'} \mathscr{L}(Q, f'). 
\end{align*}
    
Using Fact 1(a) and Fact 2(b), we obtain:
\[
\mathscr{E}(Q, f) \leq \mathscr{E}_\rho^\gamma(P, f).
\]

Using Fact 2(a) and Fact 1(b), we get:
\[
\mathscr{E}_\rho^\gamma(P, f) \leq \mathscr{E}(Q, f) + 2 L \rho + |\gamma - \gamma^*| \rho^2.
\]

Therefore:
\[
|\mathscr{E}(Q, f) - \mathscr{E}_\rho^\gamma(P, f)| \leq 2 L \rho + |\gamma - \gamma^*| \rho^2.
\]

Apply to both levels:

- For semantic loss \( h_{\text{s}}(x) := \ell_{\text{s}}(x; \vartheta) \), use \( L := L_{\text{s}}, \gamma := \lambda \), and \( \rho \) for semantic Wasserstein radius.

- For channel loss \( h_{\text{c}}(z) := \ell_{\text{c}}(s, z; \varphi) \), use \( L := L_{\text{c}}, \gamma := \gamma \), and \( \rho := \mu \) for channel Wasserstein radius.

This completes the proof of Lemma~\ref{lem:excess_risk}.

	
Finally, we provide the proof of the following proposition that was used in proving \textbf{Fact 2}(a).
    
	\begin{proposition} \label{pro1} Let Assumption~\ref{ass:lipschitz} (a) holds. For any $f \in \mathcal{L}$ and for all  $Q \in \gB(P, \rho)$, we have
	
		\begin{align*}
		\underset{P' \in \gB({P}_{\lambda}, \rho)}{\sup} \mathscr{L} (P',f) \leq \mathscr{L}(Q,f) + 2 L \rho. 
		\end{align*}
		
	\end{proposition}
	\begin{proof} Denote $P^* \defeq \underset{P' \in \gB({P}_{\lambda}, \rho)}{\argmax} \mathscr{L} (P',f)$. We have
		\begin{align}
		&\underset{P' \in \gB({P}_{\lambda}, \rho)}{\mbox{sup }} \mathscr{L} (P',f)\\ &=  \mathscr{L}(Q,f)  + \underset{P' \in \gB({P}_{\lambda}, \rho)}{\mbox{sup }} \mathscr{L} (P',f) - \mathscr{L}(Q,f) \nonumber\\ 
		& \leq \mathscr{L}(Q,f)  +  \abs*{\mathscr{L} (P^*,f) - \mathscr{L} (Q,f) }, \nonumber\\
		& \leq \mathscr{L}(Q,f)  +  L \abs*{\mathbf{E}_{u \sim P^*} \bigS{\ell(u, f)/L} - \mathbf{E}_{x \sim Q} \bigS{\ell(u,f)/L}}\nonumber \\
		& \leq  \mathscr{L}(Q,f) + L W_1  (P^*, Q)  \nonumber\\
		& \leq  \mathscr{L}(Q,f) + L \bigS{W_2  (P^*, P) + W_2(P, Q)} \label{E:triangle2}\\
		& \leq \mathscr{L}(Q,f) + L 2 \rho, \nonumber
		\end{align}
		where the  fourth line is due to the Kantorovich-Rubinstein dual representation theorem, i.e., 
		\begin{align*}
		W_1(P, Q) = \sup_{h} \BigC{ \mathbf{E}_{u \sim P} \bigS{f(u)} - \mathbf{E}_{x \sim Q} \bigS{f(u)}: \\ 
			f(\cdot) \text{ is 1-Lipschitz}}
		\end{align*}
		and the fifth line is due to $W_1  (P^*, Q) \leq W_2  (P^*, Q)$ and triangle inequality. 
	\end{proof}

\subsection{Proof of Theorem~\ref{thrm:excess_risk}}
\label{proof:thrm_excess_risk}

To analyze the generalization performance of \OurAlg under distributional shifts, we adopt a unified theoretical framework that applies to both the semantic-level and channel-level WDRO objectives. Specifically, we abstract the analysis by defining a generic input domain $\mathcal{U}$—where $\mathcal{U} = \mathcal{X}$ corresponds to semantic inputs and $\mathcal{U} = \mathcal{Z}$ corresponds to received channel signals. The loss function $\ell(u; f)$ captures either the semantic reconstruction loss $\ell_{\text{s}}(x; \vartheta)$ or the channel distortion loss $\ell_{\text{c}}(s, z; \varphi)$, depending on the level of analysis. This unified notation enables a general proof strategy using WDRO theory and empirical process tools. By bounding the deviation between the empirical and population surrogate risks, and controlling the approximation gap between surrogate and worst-case risks, we derive a robust generalization guarantee that applies to both optimization layers in the bi-level setting.

\begin{proof}
	To simplify notation, we denote $ \Phi \defeq \phi_{\gamma} \circ \mathcal{F} = \set{u \mapsto \phi_{\gamma}(u, f), f \in \mathcal{L}}$ where $\mathcal{F} = \bigC{f_{\theta}, \theta \in \Theta  \subset \mathbb{R}^d}$, which represents the composition of $\phi_{\gamma}$ with each of the loss function $f_{\theta}$ parametrized by $\theta$ belonging to the parameter class $\Theta$. 
	
	Defining $f_{P} \in \argmin_{f \in \mathcal{L}} \mathscr{L}_{\rho}^{\gamma}(P, f)$ and $ \widehat{\theta}^* \in \underset{\theta \in \Theta} {\operatorname{argmin}} \ \mathbf{E}_{u \sim \widehat{P}} \bigS{\phi_{\gamma}(u,  f_{\theta})}$ such that $\mathscr{L}_{\rho}^{\gamma}(\widehat{P}, {f}_{{\theta}^*}) = \underset{\theta \in \Theta} {\operatorname{inf}}  \BigS{\mathbf{E}_{u \sim \widehat{P}} \bigS{\phi_{\gamma}(u,  f_{\theta})}  + \gamma \rho^2 }$, we  decompose the excess risk as follows:
	
	\begin{align}
	&\mathscr{E}_{\rho}^{\gamma}(P, {f}_{\widehat{\theta}^{\varepsilon}})\\	&=\mathscr{L}_{\rho}^{\gamma}(P, {f}_{\widehat{\theta}^{\varepsilon}}) - \inf_{f \in \mathcal{L}}  \mathscr{L}_{\rho}^{\gamma}(P, f) \nonumber\\
	&= \mathscr{L}_{\rho}^{\gamma}(P, {f}_{\widehat{\theta}^{\varepsilon}}) -   \mathscr{L}_{\rho}^{\gamma}(P, f_{P})\nonumber\\
	&=  \BigS{\mathscr{L}_{\rho}^{\gamma}(P, {f}_{\widehat{\theta}^{\varepsilon}}) - \mathscr{L}_{\rho}^{\gamma}(\widehat{P}, {f}_{\widehat{\theta}^{\varepsilon}})} + \underbrace{\BigS{\mathscr{L}_{\rho}^{\gamma}(\widehat{P}, {f}_{\widehat{\theta}^{\varepsilon}}) - \mathscr{L}_{\rho}^{\gamma}(\widehat{P}, {f}_{\widehat{\theta}^{*}})} }_{\leq \varepsilon}   \nonumber\\
	& \,  + \underbrace{\BigS{\mathscr{L}_{\rho}^{\gamma}(\widehat{P}, {f}_{\widehat{\theta}^{*}}) - \mathscr{L}_{\rho}^{\gamma}(\widehat{P}, f_{P})}}_{\leq 0} + \BigS{\mathscr{L}_{\rho}^{\gamma}(\widehat{P}, f_{P}) - \mathscr{L}_{\rho}^{\gamma}(P, f_{P})} \nonumber\\
	&\leq 2 \sup_{\phi_{\gamma} \in \Phi} \abs*{ \mathbf{E}_{u \sim P}[\phi_{\gamma}(u, f_{\theta})] -   \mathbf{E}_{u \sim \widehat{P}}[\phi_{\gamma}(u, f_{\theta})] } + \varepsilon \nonumber\\
	&\leq 2 \sup_{\phi_{\gamma} \in \Phi}  \sum_{i=1}^{m} \lambda_i \abs*{\mathbf{E}_{Z_i \sim P_i}[\phi_{\gamma}(u_i, f_{\theta})] -   \mathbf{E}_{Z_i \sim \widehat{P}_{i}}[\phi_{\gamma}(u_i, f_{\theta})]}  + \varepsilon \nonumber\\
	&\leq 2   \sum_{i=1}^{m} \lambda_i \sup_{\phi_{\gamma} \in \Phi} \abs*{\mathbf{E}_{Z_i \sim P_i}[\phi_{\gamma}(u_i, f_{\theta})] -   \mathbf{E}_{Z_i \sim \widehat{P}_{i}}[\phi_{\gamma}(u_i, f_{\theta})]}  + \varepsilon \nonumber\\
    \end{align}
    \begin{align}
	& \leq  \sum_{i=1}^{m} \lambda_i  \biggS{4 \mathscr{R}_{i}(\Phi)+  2 M_{\ell} \sqrt{\frac{2 \log (2 m / \delta)}{n_i}}} + \varepsilon\, \\ &\text{with probability at least }  1-\delta, \label{unionbound1}
	\end{align}
	
	where the first inequality is due to optimization error and definition of $\widehat{\theta}^{*}$. The second inequality is due to the fact that $ \abs{\sum_{i=1}^m \lambda_i a_i } \leq \sum_{i=1}^m \lambda_i  \abs{a_i}, \forall a_i \in \mathbb{R}$ and $\lambda_i \geq 0$. The third inequality is because pushing the $\sup$ inside increases the value. For the last inequality, using  the facts that (i) $\abs{\phi_\gamma(u,f)} \leq M_{\ell}$  due to $-M_{\ell} \leq \ell(u, f) \leq \phi_{\gamma}(u, f) \leq \sup _{u \in \mathcal{U}} \ell(u, f) \leq M_{\ell}$ and (ii) the Rademacher complexity of the function class $\Phi$ defined by $\mathscr{R}_{i}(\Phi) = \mathbf{E}[\sup _{\phi_{\gamma} \in \Phi} \frac{1}{n_i} \sum_{k=1}^{n_i} \sigma_{k} \phi_\gamma(u_k,  f_{\theta})] $ where the expectation is w.r.t both $u_k \stackrel{\text { i.i.d. }}{\sim} P_i$ and i.i.d. Rademacher random variable $\sigma_{k}$ independent of $u_k, \forall k \in [n_i]$,  we have 
	
	\begin{align}
	\sup_{\phi_{\gamma} \in \Phi} & \abs*{\mathbf{E}_{Z_i \sim P_i}[\phi_{\gamma}(u_i, f_{\theta})] -   \mathbf{E}_{Z_i \sim \widehat{P}_{i}}[\phi_{\gamma}(u_i, f_{\theta})]}  \\ 
	&\geq 2 \mathscr{R}_{i}(\Phi)+   M_{\ell} \sqrt{\frac{2 \log (2 m / \delta)}{n_i}} \label{Rade1}
	\end{align}
	
	with probability $\leq \delta/m$ due to the standard symmetrization argument and McDiarmid's inequality~{\cite[Theorem~26.5]{shalev_shwartz_understanding_2014}}. Multiplying $\lambda_i$ to both sides of \cref{Rade1}, summing up the inequalities over all $i\in [n]$, and using union bound, we obtain \cref{unionbound1}.
	
	Define a stochastic process $\left(X_{\phi_{\gamma}}\right)_{\phi_{\gamma} \in \Phi}$ 
	$$
	X_{\phi_{\gamma}}:=\frac{1}{{\sqrt{n_i}}} \sum_{k=1}^{n_i} \sigma_{k} \phi_{\gamma}(u_k, f_{\theta})
	$$
	which is  zero-mean because $\mathbf{E}\left[X_{\phi_{\gamma}}\right]=0$ for all $\phi_{\gamma} \in \Phi$.  To upper-bound  $\mathscr{R}_{n}(\Phi)$, we first show that $\left(X_{\phi_{\gamma}}\right)_{\phi_{\gamma} \in \Phi}$ is a sub-Gaussian process with respect to the following pseudometric
	\begin{align}
	\left\|\phi_{\gamma}-\phi_{\gamma}^{\prime}\right\|_{\infty}  \defeq  \sup_{z \in \mathcal{Z}} \Big\lvert  {\phi_{\gamma}(u, f_{\theta}) - \phi_{\gamma}(u, f_{\theta'})}  \Big\rvert . 
	\end{align}  
	
	For any $t \in \mathbb{R}$,  using Hoeffding inequality with the fact that $\sigma_{k}, k \in [n]$, are i.i.d. bounded random variable with sub-Gaussian  parameter 1, we have
	$$
	\begin{aligned}
	&\mathbf{E}\left[\exp \left(t\left(X_{\phi_{\gamma}}-X_{\phi_{\gamma}^{\prime}}\right)\right)\right]\\ &=\mathbf{E}\left[\exp \left(\frac{t}{\sqrt{n_i}} \sum_{k=1}^{n_i} \sigma_{k}\left(\phi_{\gamma}\left(u_k, f_{\theta}\right)-\phi\left(u_k, f_{\theta'}\right)\right)\right)\right] \\
	&=\left(\mathbf{E}\left[\exp \left(\frac{t}{{\sqrt{n_i}}} \sigma_{1}\left(\phi_{\gamma}\left(u_{1}, f_{\theta}\right)-\phi_{\gamma}\left(u_{1}, f_{\theta'}\right)\right)\right)\right]\right)^{n_i} \\
	& \leq \exp \left(\frac{t^{2} \left\|\phi_{\gamma}-\phi_{\gamma}^{\prime}\right\|_{\infty}^2}{2 }\right). 
	\end{aligned}
	$$
	Then, invoking Dudley entropy integral, we have
	\begin{align} \label{E:Dudley}
	\sqrt{n_i}\, {\mathscr{R}}_{i}(\Phi) = \mathbf{E} \sup_{\phi_{\gamma} \in \Phi} X_{\phi_{\gamma}} \leq {12} \int_{0}^{\infty} \sqrt{\log \mathcal{N}\left(\Phi, \norm{\cdot}_{\infty}, \epsilon\right)} \mathrm{d} \epsilon
	\end{align}
	
	We will show that when $\theta \mapsto \ell(u,f_{\theta})$ is $L$-Lipschitz by Assumption~\ref{ass:lipschitz}, then $\theta \mapsto \phi_{\gamma}(u, f_{\theta})$ is also $L$-Lipschitz as follows. 
	\begin{align*}
	&\Big|\phi_{\gamma}(u, f_{\theta}) - \phi_{\gamma}(u, f_{\theta'})\Big| \\
	&=  \Big| \sup_{\zeta \in \mathcal{Z}} \inf_{\zeta' \in \mathcal{Z}} \Big\{ \ell(\zeta,f_{\theta})  - \gamma d(\zeta,  z)   -   \ell(\zeta', f_{\theta'}) + \gamma d(\zeta',  z)  \Big\} \Big| \\
	&\leq  \Big| \sup_{\zeta \in \mathcal{Z}} \Big\{ \ell(\zeta,f_{\theta})    -   \ell(\zeta, f_{\theta'})   \Big\} \Big| \leq \sup_{\zeta \in \mathcal{Z}}  \Big|  \ell(\zeta,f_{\theta})    -   \ell(\zeta, f_{\theta'})    \Big| \\
	& \leq L_{\theta}   \norm{\theta - \theta'}, \\
	\end{align*}
	which implies
	\begin{align*}
	\left\|\phi_{\gamma}-\phi_{\gamma}^{\prime}\right\|_{\infty} \leq  L   \norm{\theta - \theta'}.
	\end{align*}
	Therefore, by contraction principle~\cite{shalev_shwartz_understanding_2014}, we have
	\begin{align} \label{E:contraction_principle}
	\mathcal{N}\left(\Phi, \norm{\cdot}_{\infty}, \epsilon\right) \leq \mathcal{N}\left(\Theta, \norm{\cdot}, \epsilon/L_{\theta}\right).
	\end{align}
	Substituting \cref{E:contraction_principle} and \cref{E:Dudley} into \cref{unionbound1}, we obtain
	\begin{align}
	\mathscr{E}_{\rho}^{\gamma}(P, {f}_{\widehat{\theta}^{\varepsilon}}) \leq  \sum_{i=1}^{m} \lambda_i \BiggS{\frac{48 \mathscr{C}(\Theta)}{\sqrt{n_i}} + 2 M_{\ell} \sqrt{\frac{2 \log (2  m / \delta)}{n_i}}} +\varepsilon,
	\end{align}
	which will be substituted into the upper-bound in Lemma~\ref{lem:excess_risk} to complete the proof. 

This unified bound establishes a robust guarantee for both optimization layers in \OurAlg. By instantiating the general input $u$, loss $\ell(u; f)$, and parameters $(L, \lambda, \rho)$ appropriately at each level, i.e., semantic-level ($u = x$, $\ell = \ell_{\text{s}}$, $L = L_{\text{s}}$, $\lambda = \lambda$, $\rho = \rho$) and channel-level ($u = z$, $\ell = \ell_{\text{c}}$, $L = L_{\text{c}}$, $\lambda = \gamma$, $\rho = \mu$), the same analysis yields tight control over the excess risk under distributional shifts at each layer. Thus, Lemma~\ref{lem:excess_risk} ensures that minimizing the dual surrogate objectives in \OurAlg implicitly limits worst-case degradation from both semantic input perturbations and channel noise, providing theoretical justification for the bilevel robust learning formulation.

\end{proof}

\section{Hyperparameter Settings and Tuning Strategy}
\label{app:hyperparameters}

\blue{
\subsection{Specific Parameter Settings Used in Experiments}

The following parameters were used to generate the main results presented in this paper:
\begin{itemize}
    \item \textbf{Wasserstein Radii ($\rho, \mu$):} We set the semantic-level radius to $\rho=0.05$ and the channel-level radius to $\mu=0.01$ for all experiments. These values were determined through our sensitivity analysis, which demonstrated that this configuration provides an optimal balance between robustness and fidelity under adversarial conditions. The semantic radius $\rho$ is implemented practically through the use of FGSM adversarial attacks with a perturbation budget corresponding to the desired robustness level, which approximates the worst-case perturbations within the semantic Wasserstein ball.

    \item \textbf{Dual Variables ($\lambda, \gamma$):} These variables are intrinsic to the dual WDRO formulation and are not manually tuned hyperparameters. They were initialized to a value of 1.0 and subsequently updated automatically via gradient descent as part of the alternating optimization process detailed in our training algorithm. The optimization process finds their values to enforce the robustness constraints defined by $\rho$ and $\mu$.

    \item \textbf{Smoothing Parameter ($\epsilon$):} This hyperparameter controls the tightness of the log-sum-exp approximation to the supremum operation in our dual formulation. We set $\epsilon=0.1$ for all experiments. This value was found to provide an effective balance between approximation accuracy and training stability for the large-scale AI models (ViT, BERT) used in our framework.
\end{itemize}

\subsection{General Guidance for Parameter Tuning}

The effectiveness of the framework relies on selecting appropriate values for its key parameters, which govern the trade-off between average-case performance and worst-case robustness.

\begin{itemize}
    \item \textbf{Tuning the Wasserstein Radii ($\rho$ and $\mu$):} The radii, $\rho$ (semantic) and $\mu$ (channel), are the primary regularization hyperparameters. They control the level of robustness by defining the size of the ambiguity sets; smaller radii yield solutions closer to standard empirical risk minimization, while larger radii enforce greater robustness against significant distributional shifts. A standard and effective method for their selection is cross-validation, where the model is trained with several candidate values and evaluated on a validation set containing a mix of clean and perturbed data relevant to the target application.

    \item \textbf{Handling the Dual Variables ($\lambda$ and $\gamma$):} It is important to note that these variables are not manually tuned. As variables within the dual optimization problem, they are initialized (e.g., to 1.0) and subsequently updated via gradient descent during the end-to-end training process. The algorithm naturally finds their values to enforce the robustness constraints.

    \item \textbf{Tuning the Smoothing Parameter ($\epsilon$):} The parameter $\epsilon$ is a tunable hyperparameter that controls the smoothness of the approximation to the worst-case loss. Its value should be selected, typically via cross-validation, to ensure stable training convergence while keeping the approximation of the true robust objective sufficiently tight. Our experiments indicate that a value of $\epsilon=0.1$ provides a good balance for large-scale AI models.
\end{itemize}
}
\vfill

\end{document}